\documentclass{ifacconf}

\usepackage{graphicx}          
  \usepackage{amsmath,amssymb,amsfonts}

\usepackage{textcomp}
\usepackage[compatibility=false]{caption}
\usepackage{subcaption}
\usepackage{url}
\usepackage{algorithm}
\usepackage{algpseudocode}
\usepackage{tikz}
\usepackage{dcolumn}                      
\usepackage{bbm}
\usepackage{bm}
\usepackage{chngcntr}
\counterwithout{figure}{section}  
\usepackage{nicematrix}

\usepackage{comment}
\usetikzlibrary{positioning}
\definecolor{subsectioncolor}{rgb}{0.067,0.627,0.859}
\newtheorem{proposition}{Proposition}
\newtheorem{lemma}{Lemma}

\newtheorem{theorem}{Theorem}

\definecolor{mygreen}{HTML}{00F78D}

\newtheorem{proof}{Proof}

\newtheorem{definition}{Definition}

\newtheorem{remark}{Remark}
\newtheorem{eg}{Example}
\usepackage[numbers,sort&compress]{natbib}      
\begin{document}
\begin{frontmatter}

\title{Structural Sign Herdability in Temporally Switching Networks with Fixed Topology} 

\author[Indian Institute of Technology Kanpur]{Pradeep M,~Twinkle Tripathy}

\address[Indian Institute of Technology Kanpur]{Department of Electrical Engineering, 
Indian Institute of Technology Kanpur, Uttar Pradesh 208016, India.\\
Emails: \{pradeepm22, ttripathy\}@iitk.ac.in}


\begin{abstract}

This paper investigates structural sign ($\mathcal{SS}$) herdability in a special class of temporally switching networks with fixed topology. We show that when the topology of the underlying digraph remains unchanged across all snapshots, the network attains $\mathcal{SS}$ herdability even in the presence of signed or layer dilations, a condition not applicable to static networks. This reveals a fundamental structural advantage of temporal dynamics and highlights a novel mechanism through which switching can overcome the classical obstructions to herdability. To validate these conclusions, we utilize a more relaxed form of sign matching within each snapshot of the temporal network. Furthermore, we show that when all snapshots share the same underlying topology, the temporally switching network achieves $\mathcal{SS}$ herdability within just two snapshots, which is fewer than the number required for structural controllability. Several examples are included to demonstrate these results.

\end{abstract}


\end{frontmatter}

\section{INTRODUCTION}
Network controllability is a fundamental topic in network science and multi-agent systems (MAS), where controllability refers to the ability of leader agents to drive follower agents to desired states. While controllability of static networks is well studied, real-world networks are often time-varying, making their analysis significantly more challenging. Over short time intervals, time-varying networks can be modeled as sequences of static snapshots, leading to switched-system representations in which switching can influence controllability. A notable subclass is temporally switching networks, where the switching order is fixed. Such networks naturally arise in social, communication, and biological systems, motivating extensive research on their controllability using controllability Grammians ~\citep{li2017fundamental,zhang2024reachability}.

The study of structural controllability, which focuses on the network topology rather than exact edge weights~\citep{lin1974structural}, becomes more challenging in time-varying networks. The controllability of temporal networks has been studied in ~\citep{li2017fundamental}. In~\citep{zhang2024reachability}, the authors investigate structural controllability and reachability of temporally evolving networks and derive graph-theoretic bounds for the reachable subspace.

In many engineering applications, complete controllability is not necessary, since the desired operating states are often restricted to the positive orthant. For instance, fluid levels in tank systems and luminescence values in lighting and optical communication systems are inherently positive. This motivates the study of herdability, which refers to the ability of a system to reach any point in the positive orthant of \(\mathbb{R}^n\). Herdability has been widely studied in areas such as finance, cognitive science, and social networks~ \citep{devenow1996rational,welch2000herding,wermers1999mutual,raafat2009herding}. However, most existing results focus on time-invariant networks~\citep{ruf2018herdable,ruf2019herdability,meng2020leader,de2023herdability}, with only a few works~\citep{shen2025herdability} addressing herdability in time-varying networks.

Although the structural controllability of time-varying systems has been widely explored in literature, research in the direction of structural herdability is still in its early stages. While the sign pattern of the network is the basis for herdability analysis, combining it with strong structural controllability and sign controllability \citep{ruf2019herdability},  in many cases, the edge weight also affects herdability \citep{pradeep}. In our previous work \citep{pradeep}, we studied the effect of edge weights on herdability. We established a necessary and sufficient condition for $\mathcal{SS}$ herdability of an arbitrary digraph with a layered wise unisigned graph as a basis. All these studies are in time-invariant networks, while in this paper, we study the $\mathcal{SS}$ herdability of a temporally switching time-varying graph with fixed topology throughout all the snapshots, adopting from \citep{lebon2024controllability}. The main contributions of this paper are summarized as follows:
\begin{itemize}
  \item  We introduce structural sign $(\mathcal{SS})$ herdability in a special type of temporally switching networks, where the sequence of switching and topology of the network is fixed, only the edge weights vary. We account for both edge weights and sign patterns of the underlying digraph, and study the conditions for  $\mathcal{SS}$ herdability of this particular class of temporal networks.

    \item While temporal switching can improve controllability (herdability), we show that switching between systems that share the same structure and edge weights does not provide any benefit in herdability. We also show that, for a system that is not $\mathcal{SS}$ herdable to become $(\mathcal{SS})$ herdable, the switch sequence must contain at least one snapshot with a different parametric realization.

    \item Using a relaxed version of the sign-matching introduced in~\citep{pradeep}, we map the temporal segmentation in each snapshot and show that for a temporal network with fixed topology, the $\mathcal{SS}$ herdability is achieved in two snapshots. 
    
\end{itemize}

The remainder of the paper is organized as follows. In Section II, we describe the notation and reviews on $\mathcal{SS}$ herdability notions. Section III states the problem. In Section IV, we analyze the $\mathcal{SS}$ herdability of the temporal with fixed topology. Section V concludes the paper with the future direction of the work.
\vspace{-1mm}
\section{NOTATION AND BACKGROUND} \vspace{-1mm}
\label{Sec:2}
\subsection{Notation and Matrix Theory}\label{Notation}\vspace{-1mm}
The sets of real and nonnegative numbers are denoted by \(\mathbb{R}\) and \(\mathbb{R}^+\), respectively. For a vector \(k\), \([k]_i\) denotes its \(i^{\text{th}}\) entry. A vector is called unisigned if all its nonzero entries have the same sign. A matrix \(\mathcal{A}\in\mathbb{R}^{n\times n}\) is said to be nonnegative (respectively positive) if \(a_{ij}\geq0\) (respectively \(a_{ij}>0\)) for all \(i,j\in[1,n]\). The notation \({a}_{ij}\) denotes the \((i,j)^{\text{th}}\) entry of \(\mathcal{A}\), while \(\mathcal{A}_{(:,j)}\) and \(\mathcal{A}_{(i,:)}\) denote the \(j^{\text{th}}\) column and \(i^{\text{th}}\) row of \(\mathcal{A}\), respectively. The image of \(\mathcal{A}\) is defined as \(\mathrm{Im}(\mathcal{A})=\{y \mid y=\mathcal{A}v\}.\)
\subsection{Graph Theory}\label{Graph theory}
\vspace{-1mm}
Let \(\mathcal{G}=(\mathcal{V},\mathcal{E},\mathcal{A})\) denote a weighted signed digraph, where \(\mathcal{V}\) is the set of nodes and \(\mathcal{E}\subseteq \mathcal{V}\times\mathcal{V}\) is the edge set. An edge \((i,j)\in\mathcal{E}\) represents a directed connection from node \(i\) to node \(j\). The condition \(a_{ij}\neq 0\) implies that \((j,i)\in\mathcal{E}\). A walk \(w_{(i,j)}^r\) in \(\mathcal{G}(\mathcal{A},\mathcal{B})\) is a sequence of directed edges connecting an initial node $i$ to a terminal node $j$. Let \(\mathcal{W}_{(i,j)}^r\) denote the product of the edge weights along the walk \(w_{(i,j)}^r\). A path is a walk in which no node is repeated.
\vspace{-1mm}
\subsection{Linear Temporally Switching systems}
\vspace{-2mm}
A temporal network consists of an ordered sequence of \(p\) network snapshots defined over the same set of \(n\) nodes. Each snapshot \(i\) is represented by a weighted adjacency matrix \(\mathcal{A}_i\) over the interval \(T_i=[t_i-t_{i-1})\), and evolves according to
\begin{equation}\label{sys1}
\dot{\mathbf{x}}(t)=\mathcal{A}_i\mathbf{x}(t)+\mathcal{B}_i\mathbf{u}_i(t),
\end{equation}
for \(t\in [t_i-t_{i-1})\). Here, \(\mathbf{x}(t)\in\mathbb{R}^n\) denotes the system state. The matrix \(\mathcal{B}_i\) specifies the driver nodes associated with the control input \(\mathbf{u}_i(t)\in\mathbb{R}^m\). We consider temporal herdability with a single leader node in all snapshots, i.e.,\( \mathcal{B}_i=\mathcal{B}=e_1 b_i,\) where \(e_1\) is the first standard basis vector. Assuming \(\mathbf{x}(0)=\mathbf{0}\), the reachable set of system~\eqref{sys1} over \(\{T_i\}_{i=1}^p\) is given by~\citep{xie2003controllability}
{\small
\[
\Omega_{\{T_i\}}
=
\langle \mathcal{A}_p \mid \mathcal{B}_p \rangle
+
\sum_{j=2}^{p}
\left(
\prod_{i=j}^{p} e^{\mathcal{A}_iT_i}
\right)
\langle \mathcal{A}_{j-1}\mid \mathcal{B}_{j-1}\rangle .
\]
}
Here, \(\langle \mathcal{A}_p \mid \mathcal{B}_p \rangle
=
\sum_{i=0}^{n-1}\mathcal{A}_p^i\mathcal{R}(\mathcal{B}_p)\) denotes the controllable subspace of snapshot \(p\).

\begin{remark}(\citep{zhang2024reachability}) The controllability matrix $\mathcal{C}$ of the network with dynamics \eqref{sys1} is given by:
{\small
\begin{equation}\label{C-temp}
\mathcal{C}_T(\mathcal{A}_i,\mathcal{B}) = 
\big[\, C_p, \; e^{A_p T_p} C_{p-1}, \; 
\dots, \;
e^{A_p T_p} \cdots e^{A_2 tT_2} C_1 \,\big].
\end{equation}}

Where $\mathcal{C}_k$ is the controllability matrix of the pair $(\mathcal{A}_k,\mathcal{B})$
\end{remark}
\begin{proposition} {\citep{pradeep}} \label{test for H} A system is said to be herdable if and only if there exists no $\bm{y}\geq0;~\bm{y}~\in~\mathbb{R}^n$ where n is the number of rows in $\mathcal{C}(\mathcal{A},\mathcal{B})$, such that $\mathcal{C}(\mathcal{A},\mathcal{B})^{\top} \bm{y}=0$, where $\mathcal{C}(\mathcal{A},\mathcal{B})$ is the controllability matrix associated with the pair $(\mathcal{A},\mathcal{B})$. In other words,if there is no nonnegative $\bm{y}$ such that $\bm{y}~\in~\textit{Null}~(\mathcal{C}(\mathcal{A},\mathcal{B})^{\top})$, then there exists a $v>0; v\in \mathbb{R}^n$ such that $v\in \mathcal{I}m(\mathcal{C}(\mathcal{A},\mathcal{B}))$ .
\end{proposition}
    
\subsection{Signed Layered Graph and its Associated Dilations}
\subsection{Signed layered graph}\vspace{-3mm}
A signed layered graph $\mathcal{G}_s$, is a tree-like version of the digraph $\mathcal{G}(A,B)$, which shows the distances from the root node to all other nodes, including the signs between edges. Let $V_{L_p}$ is the set of nodes in layer $L_p$, where $p$ is between 1 and $n-1$ then a signed layered graph will have the following characteristics:

\begin{itemize} \item In the signed layered graph $\mathcal{G}_s$, the nodes in the $k^{th}$ layer $L_k$ are called $V_{L_k}$. The first layer $V_{L_1}$ has only the leader node. 
\item $V_{L_k}$ has nodes that can be reached in $k-1$ steps from the leader. 
\item A node $v$ can show up more than once in $V_{L_{m+1}}$ if there are many paths of the same length $m$ to it. 
\item To avoid endless layers due to cycles, we limit $\mathcal{G}_s$ to $n$ layers. This matches Remark \ref{Remark1}, as only $n$ columns matter in the controllability matrix.
  \end{itemize}
  A layered graph includes all possible paths from the starting node to any other node in the digraph. If there are no edges between two layers, then the graph is disconnected. For a detailed discussion, the reader is referred to \cite{pradeep}.
  \begin{lemma} \label{SLG}
   Let $\mathcal{G}_T(\mathcal{A}_i,\mathcal{B})$ be the digraph associated with a linear temporally switching network with dynamics described by \eqref{sys1}. Let $\mathcal{G}_{s(i)}$ be the signed layered graph associated with $(\mathcal{A}_i,\mathcal{B})$ on the interval $[t_{i-1},t_i)$, then each layer $L_d$ corresponds to the column in the controllability matrix that is associated with the pair $(\mathcal{A}_i,\mathcal{B})$ on respective snapshot.
  \end{lemma}
 

\subsubsection{Signed Dilation.} \cite{ruf2018herdable}: \label{SD para} A digraph is said to have a signed dilation if it consists of a node whose outgoing edges have different signs, as shown in Fig.~ (\ref{fig: dilation}). 

\subsubsection{Layer Dilation.}\cite{pradeep}:\label{Lyr para} A signed layered graph $\mathcal{G}s$ is said to exhibit a layer dilation in the $k^{th}$ layer if the outgoing edges from $\mathcal{V}_{L_k}$ to $(k+1)^{th}$ layer have mixed signs. For instance, the signed network in Fig.~\ref{dila}(\subref{lyr dilation}) shows a layer dilation in the layer $L_k$.
\begin{figure}[ht]
  \centering
  \begin{subfigure}{0.237\textwidth}
    \centering

\tikzset{every picture/.style={scale=0.75pt}} 

\begin{tikzpicture}[x=0.75pt,y=0.75pt,yscale=-1,xscale=1]

\draw [color={rgb, 255:red, 0; green, 0; blue, 0 }  ,draw opacity=1 ]   (141,51.46) -- (119.68,82.97) ;
\draw [shift={(118,85.46)}, rotate = 304.08] [fill={rgb, 255:red, 0; green, 0; blue, 0 }  ,fill opacity=1 ][line width=0.08]  [draw opacity=0] (10.72,-5.15) -- (0,0) -- (10.72,5.15) -- (7.12,0) -- cycle    ;
\draw [color={rgb, 255:red, 0; green, 0; blue, 0 }  ,draw opacity=1 ]   (148.33,52.52) -- (162.8,85.71) ;
\draw [shift={(164,88.46)}, rotate = 246.44] [fill={rgb, 255:red, 0; green, 0; blue, 0 }  ,fill opacity=1 ][line width=0.08]  [draw opacity=0] (10.72,-5.15) -- (0,0) -- (10.72,5.15) -- (7.12,0) -- cycle    ;
\draw  [color={rgb, 255:red, 0; green, 0; blue, 0 }  ,draw opacity=1 ][fill={rgb, 255:red, 22; green, 125; blue, 223 }  ,fill opacity=1 ] (108.87,96.35) .. controls (108.87,93.9) and (111.01,91.91) .. (113.64,91.91) .. controls (116.28,91.91) and (118.42,93.9) .. (118.42,96.35) .. controls (118.42,98.8) and (116.28,100.79) .. (113.64,100.79) .. controls (111.01,100.79) and (108.87,98.8) .. (108.87,96.35) -- cycle ;
\draw  [color={rgb, 255:red, 0; green, 0; blue, 0 }  ,draw opacity=1 ][fill={rgb, 255:red, 22; green, 125; blue, 223 }  ,fill opacity=1 ] (164.98,98.93) .. controls (164.98,96.48) and (167.12,94.49) .. (169.76,94.49) .. controls (172.39,94.49) and (174.53,96.48) .. (174.53,98.93) .. controls (174.53,101.38) and (172.39,103.36) .. (169.76,103.36) .. controls (167.12,103.36) and (164.98,101.38) .. (164.98,98.93) -- cycle ;
\draw  [color={rgb, 255:red, 0; green, 0; blue, 0 }  ,draw opacity=1 ][fill={rgb, 255:red, 22; green, 125; blue, 223 }  ,fill opacity=1 ] (139.87,41.35) .. controls (139.87,38.9) and (142.01,36.91) .. (144.64,36.91) .. controls (147.28,36.91) and (149.42,38.9) .. (149.42,41.35) .. controls (149.42,43.8) and (147.28,45.79) .. (144.64,45.79) .. controls (142.01,45.79) and (139.87,43.8) .. (139.87,41.35) -- cycle ;

\draw (92.88,90.5) node [anchor=north west][inner sep=0.75pt]  [font=\small,color={rgb, 255:red, 0; green, 0; blue, 0 }  ,opacity=1 ,rotate=-359.39]  {$2$};
\draw (158.45,30.73) node [anchor=north west][inner sep=0.75pt]  [font=\small,color={rgb, 255:red, 0; green, 0; blue, 0 }  ,opacity=1 ,rotate=-358.27]  {$1$};
\draw (180.08,90.38) node [anchor=north west][inner sep=0.75pt]  [font=\small,color={rgb, 255:red, 0; green, 0; blue, 0 }  ,opacity=1 ,rotate=-359.39]  {$3$};
\draw (113.06,57.87) node [anchor=north west][inner sep=0.75pt]  [font=\scriptsize,color={rgb, 255:red, 0; green, 0; blue, 0 }  ,opacity=1 ,rotate=-358.88]  {$+$};
\draw (163.47,58.82) node [anchor=north west][inner sep=0.75pt]  [font=\scriptsize,color={rgb, 255:red, 0; green, 0; blue, 0 }  ,opacity=1 ]  {$-$};

\end{tikzpicture}
    \caption{Signed dilation}
    \label{fig: dilation}
  \end{subfigure}
  \begin{subfigure}{0.23\textwidth}
    \centering

\tikzset{every picture/.style={scale=0.75pt}} 
   
\begin{tikzpicture}[x=0.75pt,y=0.75pt,yscale=-1,xscale=1]

\draw  [color={rgb, 255:red, 0; green, 0; blue, 0 }  ,draw opacity=0 ][fill={rgb, 255:red, 74; green, 144; blue, 226 }  ,fill opacity=0.1 ] (277.51,28.67) -- (439.93,28.67) -- (439.93,49.13) -- (277.51,49.13) -- cycle ;
\draw  [color={rgb, 255:red, 0; green, 0; blue, 0 }  ,draw opacity=0 ][fill={rgb, 255:red, 74; green, 144; blue, 226 }  ,fill opacity=0.1 ] (277.84,90.44) -- (440.26,90.44) -- (440.26,110.89) -- (277.84,110.89) -- cycle ;
\draw  [color={rgb, 255:red, 0; green, 0; blue, 0 }  ,draw opacity=1 ][fill={rgb, 255:red, 22; green, 125; blue, 223 }  ,fill opacity=1 ] (324.87,38.35) .. controls (324.87,35.9) and (327.01,33.91) .. (329.64,33.91) .. controls (332.28,33.91) and (334.42,35.9) .. (334.42,38.35) .. controls (334.42,40.8) and (332.28,42.79) .. (329.64,42.79) .. controls (327.01,42.79) and (324.87,40.8) .. (324.87,38.35) -- cycle ;
\draw  [color={rgb, 255:red, 0; green, 0; blue, 0 }  ,draw opacity=1 ][fill={rgb, 255:red, 22; green, 125; blue, 223 }  ,fill opacity=1 ] (325.29,98.85) .. controls (325.29,96.4) and (327.43,94.41) .. (330.06,94.41) .. controls (332.7,94.41) and (334.84,96.4) .. (334.84,98.85) .. controls (334.84,101.3) and (332.7,103.29) .. (330.06,103.29) .. controls (327.43,103.29) and (325.29,101.3) .. (325.29,98.85) -- cycle ;
\draw  [color={rgb, 255:red, 0; green, 0; blue, 0 }  ,draw opacity=1 ][fill={rgb, 255:red, 22; green, 125; blue, 223 }  ,fill opacity=1 ] (378.03,99.48) .. controls (378.03,97.03) and (380.17,95.05) .. (382.8,95.05) .. controls (385.44,95.05) and (387.57,97.03) .. (387.57,99.48) .. controls (387.57,101.94) and (385.44,103.92) .. (382.8,103.92) .. controls (380.17,103.92) and (378.03,101.94) .. (378.03,99.48) -- cycle ;
\draw  [color={rgb, 255:red, 0; green, 0; blue, 0 }  ,draw opacity=1 ][fill={rgb, 255:red, 22; green, 125; blue, 223 }  ,fill opacity=1 ] (376.98,38.93) .. controls (376.98,36.48) and (379.12,34.49) .. (381.76,34.49) .. controls (384.39,34.49) and (386.53,36.48) .. (386.53,38.93) .. controls (386.53,41.38) and (384.39,43.36) .. (381.76,43.36) .. controls (379.12,43.36) and (376.98,41.38) .. (376.98,38.93) -- cycle ;
\draw [color={rgb, 255:red, 208; green, 2; blue, 27 }  ,draw opacity=1 ]   (330.01,50.68) -- (330,85.46) ;
\draw [shift={(330,88.46)}, rotate = 270.02] [fill={rgb, 255:red, 208; green, 2; blue, 27 }  ,fill opacity=1 ][line width=0.08]  [draw opacity=0] (10.72,-5.15) -- (0,0) -- (10.72,5.15) -- (7.12,0) -- cycle    ;
\draw [color={rgb, 255:red, 208; green, 2; blue, 27 }  ,draw opacity=1 ][fill={rgb, 255:red, 183; green, 176; blue, 176 }  ,fill opacity=1 ]   (383,50.46) -- (382.82,89.05) ;
\draw [shift={(382.8,92.05)}, rotate = 270.27] [fill={rgb, 255:red, 208; green, 2; blue, 27 }  ,fill opacity=1 ][line width=0.08]  [draw opacity=0] (10.72,-5.15) -- (0,0) -- (10.72,5.15) -- (7.12,0) -- cycle    ;
\draw  [fill={rgb, 255:red, 189; green, 16; blue, 224 }  ,fill opacity=0.18 ] (285,59.98) .. controls (285,58.25) and (286.4,56.86) .. (288.12,56.86) -- (417.88,56.86) .. controls (419.6,56.86) and (421,58.25) .. (421,59.98) -- (421,69.34) .. controls (421,71.06) and (419.6,72.46) .. (417.88,72.46) -- (288.12,72.46) .. controls (286.4,72.46) and (285,71.06) .. (285,69.34) -- cycle ;

\draw (395.8,57.84) node [anchor=north west][inner sep=0.75pt]    {$-$};
\draw (308.08,56.83) node [anchor=north west][inner sep=0.75pt]    {$+$};
\draw (308.88,31.5) node [anchor=north west][inner sep=0.75pt]  [font=\small,color={rgb, 255:red, 0; green, 0; blue, 0 }  ,opacity=1 ,rotate=-359.39]  {$2$};
\draw (392.08,32.38) node [anchor=north west][inner sep=0.75pt]  [font=\small,color={rgb, 255:red, 0; green, 0; blue, 0 }  ,opacity=1 ,rotate=-359.39]  {$3$};
\draw (306.08,92.38) node [anchor=north west][inner sep=0.75pt]  [font=\small,color={rgb, 255:red, 0; green, 0; blue, 0 }  ,opacity=1 ,rotate=-359.39]  {$4$};
\draw (395.88,90.5) node [anchor=north west][inner sep=0.75pt]  [font=\small,color={rgb, 255:red, 0; green, 0; blue, 0 }  ,opacity=1 ,rotate=-359.39]  {$5$};
\draw (460,30.4) node [anchor=north west][inner sep=0.75pt]    {$L_{k}$};
\draw (460,90.4) node [anchor=north west][inner sep=0.75pt]    {$L_{k+1}$};

\end{tikzpicture}
    \caption{Layer dilation}
    \label{lyr dilation}
  \end{subfigure}

  \caption{Types of dilations}
  \label{dila}
\end{figure}

\vspace{-2mm}
\section{MOTIVATION AND PROBLEM STATEMENT}\label{Sec:3} 
\vspace{-1mm}
We investigate the $\mathcal{SS}$ herdability of digraphs that are not $\mathcal{SS}$ herdable in the static counterpart. In particular, tree graphs exhibiting signed dilation or layer dilation are not $\mathcal{SS}$ herdable digraphs under static conditions. To address this, we examine their $\mathcal{SS}$ herdability in the temporal setting by concatenating system matrices $\mathcal{A}_1, \mathcal{A}_2, \dots, \mathcal{A}_p$, where each $\mathcal{A}_i$ shares the same underlying topology but may take different parameter realizations within the structural framework.

To motivate our analysis, consider the digraph shown in Fig.~\ref{main di}(\subref{Temp_H}). Its static counterpart is not $\mathcal{SS}$ herdable (not herdable for any parametric realization) due to the presence of a signed dilation. It is known that the same digraph is not structurally controllable. However, the temporal version of this digraph achieves full generic rank when expanded over future snapshots with different parametric realizations. Temporal networks often enjoy a fundamental advantage over static networks, as they can achieve controllability with significantly lower control energy~\citep{li2017fundamental}. Interestingly, we show that the number of snapshots required for the digraph to become $\mathcal{SS}$ herdable is smaller than the number required to achieve full structural controllability. That is, even before the temporal network becomes fully controllable, it may already satisfy the $\mathcal{SS}$ herdability condition.

\begin{figure}[ht]
   \centering
\begin{subfigure}{0.20\textwidth}
\centering
\tikzset{every picture/.style={scale=1.32pt}} 
\begin{tikzpicture}[x=0.75pt,y=0.75pt,yscale=-1,xscale=1]

\draw    (61.69,61.31) -- (87.27,79.87) ;
\draw [shift={(75.61,71.41)}, rotate = 215.96] [fill={rgb, 255:red, 0; green, 0; blue, 0 }  ][line width=0.08]  [draw opacity=0] (5.36,-2.57) -- (0,0) -- (5.36,2.57) -- (3.56,0) -- cycle    ;
\draw  [fill={rgb, 255:red, 6; green, 80; blue, 250 }  ,fill opacity=1 ] (84.64,79.87) .. controls (84.64,78.43) and (85.82,77.25) .. (87.27,77.25) .. controls (88.72,77.25) and (89.9,78.43) .. (89.9,79.87) .. controls (89.9,81.32) and (88.72,82.49) .. (87.27,82.49) .. controls (85.82,82.49) and (84.64,81.32) .. (84.64,79.87) -- cycle ;
\draw [color={rgb, 255:red, 161; green, 27; blue, 43 }  ,draw opacity=1 ]   (61.62,37.49) .. controls (63.28,39.16) and (63.28,40.83) .. (61.61,42.49) -- (61.61,43.97) -- (61.59,51.97) ;
\draw [shift={(61.59,54.97)}, rotate = 270.1] [fill={rgb, 255:red, 161; green, 27; blue, 43 }  ,fill opacity=1 ][line width=0.08]  [draw opacity=0] (7.14,-3.43) -- (0,0) -- (7.14,3.43) -- (4.74,0) -- cycle    ;
\draw  [fill={rgb, 255:red, 239; green, 43; blue, 67 }  ,fill opacity=1 ] (58.99,34.87) .. controls (58.99,33.43) and (60.17,32.26) .. (61.62,32.26) .. controls (63.07,32.26) and (64.25,33.43) .. (64.25,34.87) .. controls (64.25,36.32) and (63.07,37.49) .. (61.62,37.49) .. controls (60.17,37.49) and (58.99,36.32) .. (58.99,34.87) -- cycle ;
\draw    (61.69,61.31) -- (35.8,79.05) ;
\draw [shift={(47.59,70.97)}, rotate = 325.59] [fill={rgb, 255:red, 0; green, 0; blue, 0 }  ][line width=0.08]  [draw opacity=0] (5.36,-2.57) -- (0,0) -- (5.36,2.57) -- (3.56,0) -- cycle    ;
\draw    (61.71,61.48) -- (62.03,85.26) ;
\draw [shift={(61.89,74.77)}, rotate = 269.22] [fill={rgb, 255:red, 0; green, 0; blue, 0 }  ][line width=0.08]  [draw opacity=0] (5.36,-2.57) -- (0,0) -- (5.36,2.57) -- (3.56,0) -- cycle    ;
\draw  [color={rgb, 255:red, 64; green, 95; blue, 17 }  ,draw opacity=1 ][fill={rgb, 255:red, 85; green, 211; blue, 33 }  ,fill opacity=1 ] (61.59,55.78) .. controls (63.06,55.75) and (64.28,57) .. (64.31,58.57) .. controls (64.35,60.14) and (63.18,61.44) .. (61.71,61.48) .. controls (60.23,61.51) and (59.01,60.26) .. (58.98,58.68) .. controls (58.95,57.11) and (60.12,55.81) .. (61.59,55.78) -- cycle ;
\draw  [fill={rgb, 255:red, 6; green, 80; blue, 250 }  ,fill opacity=1 ] (33.17,79.05) .. controls (33.17,77.6) and (34.34,76.43) .. (35.8,76.43) .. controls (37.25,76.43) and (38.43,77.6) .. (38.43,79.05) .. controls (38.43,80.49) and (37.25,81.66) .. (35.8,81.66) .. controls (34.34,81.66) and (33.17,80.49) .. (33.17,79.05) -- cycle ;
\draw  [fill={rgb, 255:red, 6; green, 80; blue, 250 }  ,fill opacity=1 ] (59.4,85.26) .. controls (59.4,83.81) and (60.57,82.64) .. (62.03,82.64) .. controls (63.48,82.64) and (64.66,83.81) .. (64.66,85.26) .. controls (64.66,86.7) and (63.48,87.88) .. (62.03,87.88) .. controls (60.57,87.88) and (59.4,86.7) .. (59.4,85.26) -- cycle ;

\draw (66.74,51.77) node [anchor=north west][inner sep=0.75pt]  [font=\tiny]  {$1$};
\draw (24.47,73.24) node [anchor=north west][inner sep=0.75pt]  [font=\tiny]  {$2$};
\draw (65.52,82.93) node [anchor=north west][inner sep=0.75pt]  [font=\tiny]  {$3$};
\draw (90.7,76.36) node [anchor=north west][inner sep=0.75pt]  [font=\tiny]  {$4$};
\draw (57.85,23.27) node [anchor=north west][inner sep=0.75pt]  [font=\tiny]  {$u(t)$};
\draw (65.72,40.3) node [anchor=north west][inner sep=0.75pt]  [font=\tiny,color={rgb, 255:red, 128; green, 128; blue, 128 }  ,opacity=1 ]  {$1$};
\draw (42.05,59.97) node [anchor=north west][inner sep=0.75pt]  [font=\tiny,color={rgb, 255:red, 128; green, 128; blue, 128 }  ,opacity=1 ]  {$1$};
\draw (80.07,65.33) node [anchor=north west][inner sep=0.75pt]  [font=\tiny,color={rgb, 255:red, 128; green, 128; blue, 128 }  ,opacity=1 ]  {$-1$};
\draw (49.07,75) node [anchor=north west][inner sep=0.75pt]  [font=\tiny,color={rgb, 255:red, 128; green, 128; blue, 128 }  ,opacity=1 ]  {$-1$};

\end{tikzpicture}
\caption{A digraph $\mathcal{G}_1(\mathcal{A,B})$}
\label{temp_1}

\end{subfigure}
\hspace{2mm}
\begin{subfigure}{0.2\textwidth}
\centering   
\tikzset{every picture/.style={scale=0.95pt}} 

\begin{tikzpicture}[x=0.75pt,y=0.75pt,yscale=-1,xscale=1]

\draw  [fill={rgb, 255:red, 6; green, 80; blue, 250 }  ,fill opacity=1 ] (241.28,330.03) .. controls (241.28,328.58) and (242.46,327.41) .. (243.91,327.41) .. controls (245.37,327.41) and (246.55,328.58) .. (246.55,330.03) .. controls (246.55,331.47) and (245.37,332.64) .. (243.91,332.64) .. controls (242.46,332.64) and (241.28,331.47) .. (241.28,330.03) -- cycle ;
\draw [color={rgb, 255:red, 161; green, 27; blue, 43 }  ,draw opacity=1 ]   (222.62,278.35) .. controls (224.28,280.02) and (224.28,281.68) .. (222.61,283.35) -- (222.61,284.83) -- (222.59,292.83) ;
\draw [shift={(222.59,295.83)}, rotate = 270.1] [fill={rgb, 255:red, 161; green, 27; blue, 43 }  ,fill opacity=1 ][line width=0.08]  [draw opacity=0] (7.14,-3.43) -- (0,0) -- (7.14,3.43) -- (4.74,0) -- cycle    ;
\draw  [fill={rgb, 255:red, 239; green, 43; blue, 67 }  ,fill opacity=1 ] (219.99,275.73) .. controls (219.99,274.29) and (221.17,273.11) .. (222.62,273.11) .. controls (224.07,273.11) and (225.25,274.29) .. (225.25,275.73) .. controls (225.25,277.18) and (224.07,278.35) .. (222.62,278.35) .. controls (221.17,278.35) and (219.99,277.18) .. (219.99,275.73) -- cycle ;
\draw  [color={rgb, 255:red, 64; green, 95; blue, 17 }  ,draw opacity=1 ][fill={rgb, 255:red, 85; green, 211; blue, 33 }  ,fill opacity=1 ] (222.59,296.63) .. controls (224.06,296.6) and (225.28,297.85) .. (225.31,299.43) .. controls (225.35,301) and (224.18,302.3) .. (222.71,302.33) .. controls (221.23,302.36) and (220.01,301.11) .. (219.98,299.54) .. controls (219.95,297.97) and (221.12,296.67) .. (222.59,296.63) -- cycle ;
\draw  [fill={rgb, 255:red, 6; green, 80; blue, 250 }  ,fill opacity=1 ] (202.4,330.66) .. controls (202.4,329.21) and (203.58,328.04) .. (205.03,328.04) .. controls (206.49,328.04) and (207.67,329.21) .. (207.67,330.66) .. controls (207.67,332.1) and (206.49,333.28) .. (205.03,333.28) .. controls (203.58,333.28) and (202.4,332.1) .. (202.4,330.66) -- cycle ;
\draw  [fill={rgb, 255:red, 6; green, 80; blue, 250 }  ,fill opacity=1 ] (182.37,358.62) .. controls (182.37,357.17) and (183.55,356) .. (185,356) .. controls (186.45,356) and (187.63,357.17) .. (187.63,358.62) .. controls (187.63,360.06) and (186.45,361.23) .. (185,361.23) .. controls (183.55,361.23) and (182.37,360.06) .. (182.37,358.62) -- cycle ;
\draw    (203.47,333.02) -- (221.38,357.37) ;
\draw [shift={(213.96,347.29)}, rotate = 233.66] [fill={rgb, 255:red, 0; green, 0; blue, 0 }  ][line width=0.08]  [draw opacity=0] (7.14,-3.43) -- (0,0) -- (7.14,3.43) -- (4.74,0) -- cycle    ;
\draw    (203.47,333.02) -- (186,356) ;
\draw [shift={(193.16,346.58)}, rotate = 307.24] [fill={rgb, 255:red, 0; green, 0; blue, 0 }  ][line width=0.08]  [draw opacity=0] (7.14,-3.43) -- (0,0) -- (7.14,3.43) -- (4.74,0) -- cycle    ;
\draw    (244.91,332.64) -- (264.52,357.66) ;
\draw [shift={(256.32,347.2)}, rotate = 231.91] [fill={rgb, 255:red, 0; green, 0; blue, 0 }  ][line width=0.08]  [draw opacity=0] (7.14,-3.43) -- (0,0) -- (7.14,3.43) -- (4.74,0) -- cycle    ;
\draw    (222.51,302.84) -- (242.81,327.94) ;
\draw [shift={(234.3,317.42)}, rotate = 231.04] [fill={rgb, 255:red, 0; green, 0; blue, 0 }  ][line width=0.08]  [draw opacity=0] (7.14,-3.43) -- (0,0) -- (7.14,3.43) -- (4.74,0) -- cycle    ;
\draw    (222.51,302.84) -- (205.67,327.66) ;
\draw [shift={(212.63,317.4)}, rotate = 304.17] [fill={rgb, 255:red, 0; green, 0; blue, 0 }  ][line width=0.08]  [draw opacity=0] (7.14,-3.43) -- (0,0) -- (7.14,3.43) -- (4.74,0) -- cycle    ;
\draw  [fill={rgb, 255:red, 6; green, 80; blue, 250 }  ,fill opacity=1 ] (220.75,358.99) .. controls (220.75,357.55) and (221.93,356.37) .. (223.38,356.37) .. controls (224.83,356.37) and (226.01,357.55) .. (226.01,358.99) .. controls (226.01,360.44) and (224.83,361.61) .. (223.38,361.61) .. controls (221.93,361.61) and (220.75,360.44) .. (220.75,358.99) -- cycle ;
\draw  [fill={rgb, 255:red, 6; green, 80; blue, 250 }  ,fill opacity=1 ] (262.89,360.28) .. controls (262.89,358.83) and (264.07,357.66) .. (265.52,357.66) .. controls (266.98,357.66) and (268.16,358.83) .. (268.16,360.28) .. controls (268.16,361.72) and (266.98,362.89) .. (265.52,362.89) .. controls (264.07,362.89) and (262.89,361.72) .. (262.89,360.28) -- cycle ;

\draw (227.74,292.63) node [anchor=north west][inner sep=0.75pt]  [font=\tiny]  {$1$};
\draw (192.14,326.09) node [anchor=north west][inner sep=0.75pt]  [font=\tiny]  {$2$};
\draw (252.52,325.79) node [anchor=north west][inner sep=0.75pt]  [font=\tiny]  {$3$};
\draw (170.36,356.55) node [anchor=north west][inner sep=0.75pt]  [font=\tiny]  {$4$};
\draw (218.85,260.79) node [anchor=north west][inner sep=0.75pt]  [font=\tiny]  {$u( t)$};
\draw (180.3,336.16) node [anchor=north west][inner sep=0.75pt]  [font=\scriptsize]  {$+$};
\draw (215.96,334.82) node [anchor=north west][inner sep=0.75pt]  [font=\scriptsize]  {$-$};
\draw (232.9,357.15) node [anchor=north west][inner sep=0.75pt]  [font=\tiny]  {$5$};
\draw (275.9,356.91) node [anchor=north west][inner sep=0.75pt]  [font=\tiny]  {$6$};
\draw (209.81,279.45) node [anchor=north west][inner sep=0.75pt]  [font=\scriptsize]  {$+$};
\draw (202.39,306.04) node [anchor=north west][inner sep=0.75pt]  [font=\scriptsize]  {$+$};
\draw (240.43,307.73) node [anchor=north west][inner sep=0.75pt]  [font=\scriptsize]  {$-$};
\draw (261.05,334.04) node [anchor=north west][inner sep=0.75pt]  [font=\scriptsize]  {$+$};
\end{tikzpicture}
    \caption{A digraph $\mathcal{G}_2(\mathcal{A,B}) $}
    \label{Temp_H}
\end{subfigure}
\caption{Temporally switching digraphs that are not $\mathcal{SS}$ herdable in static version}
\label{main di}    
\end{figure}
 Th controllability matrix $\mathcal{C}(\mathcal{A,B})$ associated with the digraph given in Fig.\ref{main di}(\subref{temp_1}) is given below
{\small
\[
 \mathcal{C}(\mathcal{A,B})=\begin{bmatrix}
0 & 0 & 0 & 0 \\
0 & 1 & 0 & 0 \\
0& -1 & 0 & 0 \\
0 & -1 & 0 & 0
\end{bmatrix};
\quad
\mathcal{C}_T(\mathcal{A}_i,\mathcal{B}) =
\begin{bmatrix}
1 & 0 & 0 & 0 &1 & 0 & 0 & 0\\
0 & 1 & 0 & 0 &T_2 & 1 & 0 & 0\\
0 & -1 & 0 & 0 &-T_2 & -1 & 0 & 0\\
0 & -1& 0 & 0 &-T_2 & -1 & 0 & 0\\
\end{bmatrix}
\] }
The digraph is not \(\mathcal{SS}\) herdable for any parametric realization of \((\mathcal{A},\mathcal{B})\). To further examine this, we consider its temporal version with two snapshots over \([t_0,t_2)\), both having identical structure and edge weights. The corresponding controllability matrix \(\mathcal{C}_T(\mathcal{A}_i,\mathcal{B})\) is shown for \(t\in[t_0,t_2)\).

From \(\mathcal{C}_T\), the rows corresponding to nodes \(\{2,3,4\}\) are linearly dependent due to the presence of a signed dilation. Hence, adding more temporal snapshots with the same structure and edge weights cannot make the system fully herdable. However, the same network can be \(\mathcal{SS}\) herdable under a different parametric realization. We first introduce the following definitions.

\begin{definition}
A temporally switching network specified by the pairs $(\mathcal{A}_i,\mathcal{B})$, $i \in \{1,2,\ldots,p\}$, is said to be {structurally similar} to another network described by $(\bar{\mathcal{A}}_i,\bar{\mathcal{B}})$ if the following conditions are satisfied:
\begin{enumerate}
    \item For each snapshot $i$, the matrix $\bar{\mathcal{A}}_i$ shares the same zero-nonzero structure and sign pattern as $\mathcal{A}_i$.
    \item Both networks consist of the same number of snapshots.
\end{enumerate}
\end{definition}

\begin{definition}
Consider a temporally switching network governed by the dynamics~\eqref{sys1}. The state $x_i$ of node $i$ is said to be herdable on the interval $[t_0,t_f)$ if, for the given temporal sequence, and for every initial condition $x(t_0)\in\mathbb{R}^n$ and every threshold $h>0$, there exists a piecewise continuous input $u:\,[t_0,t_f)\to\mathbb{R}^m$ such that the solution satisfies \(x_i(t_f)\ge h\).
\end{definition}

The temporally switching network described by~\eqref{sys1} is said to be completely herdable on $[t_0,t_f)$ if, for every initial condition $x(t_0)\in\mathbb{R}^n$ and every $h>0$, there exists a piecewise continuous input $u:\,[t_0,t_f)\to\mathbb{R}^m$ such that \( x_i(t_f)\ge h\quad \text{for all nodes } i=\{1,\dots,n,\)\} i.e., every node is simultaneously herdable on $[t_0,t_f)$.

\begin{definition}
A temporally switching digraph \(\mathcal{G}_T(\mathcal{A}_k,\mathcal{B})\), where \(k \in \{1,\dots,p\}\) and \(p\) denotes the number of snapshots, is said to be structurally sign (\(\mathcal{SS}\) herdable) if there exists a temporally switching network \(\mathcal{G}_T(\mathcal{\bar{A}}_k,\mathcal{\bar{B}})\) consistent with the sign pattern and topology prescribed by \((\mathcal{A}_k,\mathcal{B})\) such that the corresponding temporal switching system is completely herdable.
\end{definition}

Consider a multi-agent network modelled by a directed weighted signed graph $\mathcal{G}(\mathcal{A},\mathcal{B})$. Let \(\mathcal{V}=\{1,\dots,n\}\) is the set of agents, \(\mathcal{V}_L=\{1\}\subseteq \mathcal{V}\) is the leader node, and \(\mathcal{V}_f=V\setminus \{1\}\) is the set of followers. \(x_i(t)\in\mathbb{R}\) denotes the state of agent \(i\), and \(u(t)\) represents the control input applied to the designated leader (node \(1\)). Let $\mathcal{G}_T(\mathcal{A}_k,\mathcal{B})$ be the temporally switching digraph with \(p\) snapshots, i.e, $k=\{1,2,\cdots p\}$. We assume that the sign pattern and topology of the digraph $\mathcal{G}_T(\mathcal{A}_k,\mathcal{B})$ remain unchanged across the snapshots, but the magnitude of edge weights can vary without changing the sign pattern, implying that \(a_{ij}\in\mathbb{R}^+\). 

Furthermore, the network is assumed to be input–connected throughout the snapshot, meaning that node~1 receives the external input and all the followers are influenced by the leader in every snapshot. In addition, node~1 evolves independently of the follower nodes, i.e., its state is unaffected by their dynamics.
 Consequently, the input matrix is identical across all snapshots and is given by  
\(
\mathcal{B}_i = \mathcal{B} = \begin{bmatrix} b_1 & 0 & 0 & 0 & \cdots \end{bmatrix}^\top \in \mathbb{R}^n,~ \forall\, i \in \{1,2,\dots,m\}
\).

\begin{remark} \label{Remark1}
    Consider a pair \( (\mathcal{A}_i, \mathcal{B}) \), where \( \mathcal{A}_i \in \mathbb{R}^{n \times n} \) and \( \mathcal{B} \in \mathbb{R}^n \) , $i\in \{1,\dots,p\}$. Assume that node 1 is the leader for all $t\in [t_{i-1},t_i)$. Then, the controllability matrix of the pair \( (\mathcal{A}_i, \mathcal{B}_i) \) is given by
    \begin{equation} \label{Control_mat}
\mathcal{C}_i(\mathcal{A}_i, \mathcal{B}) = \begin{bmatrix}
~~\mathcal{B} ~|~\Psi^i_1~|~\Psi^i_2~|....~|\Psi^i_{n-1}
\end{bmatrix}
\end{equation}

where  \( [\Psi^i_k] =[\mathcal{A}_i^k\mathcal{B}]\), $k\in[0,n-1]$. Here, $b_ 1$ is the strength of the input signal received by the leader node. Without loss of generality, the sign of $b_1$ is assumed to be positive for the remainder of this study. If the $j^{th}$ entry of $[\Psi^i_k]$, $[\Psi_k^i]_j$ is not equal to $0$, then there exists at least one path from the leader to node $j$ of length $k$ at $t\in [t_{i-1},t_i)$.

\end{remark}

\begin{theorem}\label{image of C-tmp}
    
A linear temporally switching system \eqref{sys1} is completely herdable on time interval $[t_0,t_m)$ if and only if there exists a $v>0$ such that the $v$ belongs to the range space of the controllability matrix $v\in \mathcal{I}m(\mathcal{C}_T(\mathcal{A}_i,\mathcal{B}))$ given in \eqref{C-temp}.
\end{theorem}    
\begin{proof}
   Sufficiency: We prove the claim by contradiction. Suppose that the controllability matrix  $\mathcal{C}_T$ in~\eqref{C-temp} does not have a positive image, yet the temporally switching system is herdable. Since $\mathcal{C}_T$ does not admit a positive image, there must exist two rows $\mathcal{C}_T(i,:)$ and $\mathcal{C}_T(j,:)$, corresponding to nodes $i$ and $j$, that are non-negatively linearly dependent. Hence, the system cannot generate a strictly positive state in both coordinates. In particular, we have \( \mathcal{C}_T(i,:) = -\beta\, \mathcal{C}_T(j,:)\) for some \(\beta > 0\). This implies the existence of a nonnegative vector $y \in \mathbb{R}^{n}$, such that \( \mathcal{C}_T^\top y = 0\). By Proposition~\ref{test for H}, this condition implies that the system is not herdable, which contradicts our assumption.

Necessity: Suppose that, for the given temporal sequence, the system \eqref{sys1} is herdable. Then, for any initial condition, there exists an input $u_i(t)$ defined on each interval $[t_{i-1},t_i)$ such that every state $x_i$ can be driven across a positive threshold $h>0$. The existence of a suitable input $u(t)$ for every state such that $x_i(t_f)\geq h>0$. This is true for all the state $x_i; i \in \{1,2,\dots n \}$. For a time invariant case this equivalent to the existence of a vector $x$ (constructed from $u(t)$) such that \(\mathcal{C}_T x = v > 0,\) that is, there exists a strictly positive vector $\,v \in \mathcal{I}m(\mathcal{C}_T(\mathcal{A}_i,\mathcal{B}))$. \hfill \qed
\end{proof}

The controllability matrix analysis for \(\mathcal{SS}\) herdability in temporally switching networks becomes considerably more complex than in the static case, particularly for large graphs. For illustration, consider the digraph in Fig.~\ref{main di}(\subref{Temp_H}), whose controllability matrix for two snapshots is given below.
\[\mathcal{C}_{T} (\mathcal{A}_i,\mathcal{B}) =
\left[
\begin{smallmatrix}
1 & 0 & 0 & 0 & 0 & 0 & 1 & 0 & 0 & 0 & 0 & 0 \\
0 & d_{21} & 0 & 0 & 0 & 0 & \sigma_1 & a_{21} & 0 & 0 & 0 & 0 \\
0 & d_{31} & 0 & 0 & 0 & 0 & \sigma_2 & a_{31} & 0 & 0 & 0 & 0 \\
0 & 0 & d_{21} d_{42} & 0 & 0 & 0 & \sigma_3 & \sigma_4 & a_{21} a_{42} & 0 & 0 & 0 \\
0 & 0 & d_{21} d_{52} & 0 & 0 & 0 & \sigma_5 & \sigma_6 &a_{21} a_{52} & 0 & 0 & 0 \\
0 & 0 & -d_{31} d_{63} & 0 & 0 & 0 & -\sigma_7 & -\sigma_8 & -a_{31} a_{63} & 0 & 0 & 0
\end{smallmatrix}
\right]
\]
The auxiliary coefficients $\sigma_i$ in the controllability matrix are defined as:
\[
\begin{aligned}
\sigma_1 &= d_{21} T_2, 
& \sigma_2 &= d_{31} T_2, 
& \sigma_3 &= \frac{1}{2} d_{21} d_{42} T_2^2, \\[1mm]
\sigma_4 &= a_{21} d_{42} T_2, 
& \sigma_5 &= \frac{1}{2} d_{21} d_{52} T_2^2, 
& \sigma_6 &= a_{21} d_{52} T_2, \\[1mm]
\sigma_7 &= \frac{1}{2} d_{31} d_{63} T_2^2, 
& \sigma_8 &= a_{31} d_{63} T_2.
\end{aligned}
\]
For such cases, graph-theoretic conditions are preferable to conventional algebraic methods due to their analytical convenience. In the following, we develop graph-theoretic conditions for the temporal network discussed above.
   \begin{proposition}
  Let $\mathcal{G}(\mathcal{A},\mathcal{B})$ be a digraph that is not $\mathcal{SS}$ herdable. Let $\mathcal{G}_T(\mathcal{A}_i,\mathcal{B})$ be the digraph associated with a linear temporally switching network same structure as $\mathcal{G}(\mathcal{A},\mathcal{B})$, with dynamics described by \eqref{sys1}. Let $\mathcal{G}_{s(i)}$ be the signed layered graph associated with $(\mathcal{A}_i,\mathcal{B})$ on the interval $[t_{i-1},t_i)$.  Then \(\mathcal{G}_T\) is not herdable over the interval \([t_0, t_m)\) when all snapshots share the same structure and identical edge weights.
  \end{proposition}
\begin{proof}
    Consider the controllability matrix associated with the temporally switching digraph in \eqref{C-temp}, in which the columns can be grouped into $p$ blocks such that each block corresponds to a respective switching. The matrix can be expanded as follows:
\begin{equation}\label{c-expansion}
\begin{aligned}
    \mathcal{C} = \big[\, 
    & \langle \mathcal{B} ~\big|~ \Psi^p_1 ~\big|~ \Psi^p_2 ~\big|~ .. ~\big|~ \Psi^p_{n-1} \rangle, \\
    & e^{\mathcal{A}_p  T_p}~ \langle \mathcal{B} ~\big|~ \Psi^{p-1}_1 ~\big|~ \Psi^{p-1}_2 ~\big|~ .. ~\big|~ \Psi^{p-1}_{n-1}\rangle, \\
    & .., e^{\mathcal{A}_p  T_p} \cdot\cdot~  e^{\mathcal{A}_2 T_2} ~\langle \mathcal{B} 
      ~\big|~ \Psi^1_1 ~\big|~ \Psi^1_2 ~\big|~ .. ~\big|~ \Psi^1_{n-1}\rangle \,\big].
\end{aligned}
\end{equation}
is formed by concatenating block columns $\Psi^i_j \in \mathbb{R}^{n\times m}$ defined as
 \( \Psi^i_j = A_i^{\,j-1} B_i, \qquad i = \{1,\dots,p\},\; j =\{1,\dots,n-1\}\) and $T_p=[t_{p}-t_{p-1})$ . This again simplifies to the following 
\[ \begin{aligned}
   \mathcal{C} = \big[\, 
    & \langle \mathcal{B} ~\big|~ \Psi^p_1 ~\big|~ \Psi^p_2 ~\big|~ \dots ~\big|~ \Psi^p_{n-1} \rangle, \\
    &  \langle e^{\mathcal{A}_p T_p}~\mathcal{B} ~\big|~ e^{\mathcal{A}_p T_p}~\Psi^{p-1}_1  ~\big|~ \dots ~\big|~ e^{\mathcal{A}_p T_p}~\Psi^{p-1}_{n-1}\rangle, \\
    & \dots,   ~\langle e^{\mathcal{A}_p T_p} \cdots e^{\mathcal{A}_2 T_2}\mathcal{B} 
      ~\big|~ e^{\mathcal{A}_p T_p} \cdots e^{\mathcal{A}_2 T_2}~\Psi^1_1 ~ \dots \\ &~\big|~e^{\mathcal{A}_p T_p} \cdots e^{\mathcal{A}_2 T_2} \Psi^1_{n-1}\rangle \,\big].
\end{aligned}
\]

The columns of \(\mathcal{C}_T\) are polynomial combinations of the terms arising from the Taylor expansion of \(e^{\mathcal{A}_i t}\). Since \(e^{\mathcal{A}_iT_i}\) is invertible for all \(i\), successive multiplications by these matrices do not alter the rank of \(\mathcal{C}_T\). Without loss of generality, let the first \(n\) columns of \(\mathcal{C}_T\) correspond to the realization \((\mathcal{A}_p,\mathcal{B})\). The remaining columns are linear combinations of the columns of the controllability matrices \(\mathcal{C}_k\), associated with the realizations \((\mathcal{A}_k,\mathcal{B})\), \(k \in \{1,..,p-1\}\).

Since all realizations \((\mathcal{A}_k,\mathcal{B})\) are identical to \((\mathcal{A}_p,\mathcal{B})\), the corresponding controllability matrices have the same sign structure. If \(\mathcal{C}_p\) is not herdable, then by Proposition~\ref{test for H}, it contains nonnegatively linearly dependent rows. This dependency therefore persists in \(\mathcal{C}_T\), since the additional columns generated through temporal switching are only linear combinations of the existing columns. Hence, by Proposition~\ref{test for H}, the temporally switching network remains unherdable. \hfill \qed

\end{proof}

For illustration, reconsider the example in Fig.~\ref{main di}(\subref{temp_1}). As discussed earlier, the controllability matrix corresponding to the temporal sequence \(\{\mathcal{A}_1,\mathcal{A}_2\}\) with \(\mathcal{A}_1=\mathcal{A}_2\) contains nonnegative linear dependencies among its rows. Hence, temporal switching does not improve herdability, and the system remains not \(\mathcal{SS}\)-herdable. Now consider the following temporal sequence with the same network structure but different parametric realizations:
\begin{figure}[ht]
\small
\[
\mathcal{A}_1 =
\begin{bmatrix}
0 & 0 & 0 & 0 \\
1 & 0 & 0 & 0 \\
-1 & 0 & 0 & 0 \\
-1 & 0 & 0 & 0
\end{bmatrix};
\quad
\mathcal{A}_2 =
\begin{bmatrix}
0 & 0 & 0 & 0 \\
2 & 0 & 0 & 0 \\
-3 & 0 & 0 & 0 \\
-4 & 0 & 0 & 0
\end{bmatrix};
\quad
\mathcal{B}=
\begin{bmatrix}
1  \\
0  \\
0  \\
0 
\end{bmatrix}
\]
\end{figure}

The corresponding controllability matrix is given below:
{\small
\[
\mathcal{C}_T(\mathcal{A}_i,\mathcal{B}) =
\begin{bNiceMatrix}
1 & 0 & 0 & 0 &1 & 0 & 0 & 0\\
\rowcolor{gray!15}0 & 2 & 0 & 0 &2T_2 & 1 & 0 & 0\\
\rowcolor{gray!15}0 & -3 & 0 & 0 &-3T_2 & -1 & 0 & 0\\
\rowcolor{gray!15}0 & -4& 0 & 0 &-4T_2 & -1 & 0 & 0\\
\end{bNiceMatrix}
\] }
From the above matrix, it is evident that for $T_2 > 0$, no rows of $\mathcal{C}_T$ are non–negatively linearly dependent. Hence, the temporally switching system with this realization is $\mathcal{SS}$ herdable. In the following section, we introduce path–sign matching, a relaxed variant of the sign-matching concept in~\cite{pradeep}, and demonstrate how nodes become matched across the corresponding snapshots.

\section{Structural Sign Herdability over Fixed Temporal Topologies}\label{S-4}
In this section, we introduce a relaxed form of sign matching, termed path–sign matching, which facilitates identifying the nodes that become herdable through the temporal evolution of the network across snapshots.
 
  \subsection{Path sign matching in time-varying networks}\vspace{-2mm}

We use a relaxed version of sign matching \citep{pradeep} of nodes in a signed layered graph associated with the digraph, to show how the temporal segmentation evolves with each snapshot.

\begin{definition} \label{path sign}
Let \(\mathcal{G}_T(\mathcal{A}_i,\mathcal{B})\) be the digraph associated with a linear temporally switching network consisting of \(p\) snapshots, and let \(\mathcal{G}_{s(i)}\) denote the signed layered graph associated with \((\mathcal{A}_i,\mathcal{B})\) on the interval \([t_{i-1},t_i)\). Then the following definitions hold:
\begin{itemize}
    \item A set of nodes \(S \subseteq \mathcal{V}_{L_d}\) in layer \(L_d\) is said to be \emph{path-sign-matched} in snapshot \(i\) if, for every node in \(S\), the sign of the sum of the products of all path weights from the leader node to the corresponding node in $S$ with length \(d\) is identical.

    \item The complementary set \(\bar{S} = \mathcal{V}_{L_d}\setminus S\) is said to be \emph{path-sign-matched in a subsequent snapshot} if, for every node in \(\bar{S}\), the sign of the sum of the products of all path weights from the leader node along all paths of length \(d\) is identical, while the corresponding magnitude differs from that in snapshot \(i\) due to a different parametric realization.
\end{itemize}

Moreover, the common sign associated with the nodes in \(S\) is opposite to the common sign associated with the nodes in \(\bar{S}\).
\end{definition}

As an illustrative example, consider the digraph in Fig.~\ref{total SM}(\subref{SM1}). Since the graph is a tree, its signed layered representation $\mathcal{G}_s$ is identical to the original digraph. In $\mathcal{G}_s$, the leader node~1 is in layer $L_1$, and each remaining node is assigned to a layer according to the length of its walk from the leader.

It is well known that in the signed layered graph of an input–connected network, every node in each layer admits either a positive or a negative walk from the leader. Consequently, if a node is not path-sign–matched with the leader in one snapshot, it can get path-sign-matched with the alternate sign in a subsequent snapshot. This observation forms the core intuition underlying our analysis.

In Fig.~\ref{total SM}(\subref{SM1}), during the first snapshot, nodes with a negative walk in $L_2$ and a positive walk in $L_3$ are path-matched. In the second snapshot, the remaining nodes are sign-matched within their respective layers, as shown in the figure. Notably, the nodes matched in the first snapshot are already herded by the time the second snapshot occurs.
\vspace{-1mm}
\begin{figure}[ht]
  \centering
   \begin{subfigure}{0.19\textwidth} 
    \centering

\tikzset{every picture/.style={line width=0.75pt}} 

\begin{tikzpicture}[x=0.75pt,y=0.75pt,yscale=-1,xscale=1]

\draw  [draw opacity=0][fill={rgb, 255:red, 74; green, 144; blue, 226 }  ,fill opacity=0.21 ] (308.83,391) .. controls (312.92,385.18) and (324.33,393.25) .. (341.58,392) .. controls (358.83,390.75) and (361.08,384.5) .. (372.33,389.5) .. controls (383.58,394.5) and (368.76,401.28) .. (356.67,403.39) .. controls (344.58,405.5) and (330.83,405.75) .. (321.17,404.14) .. controls (311.51,402.53) and (304.75,396.82) .. (308.83,391) -- cycle ;
\draw  [draw opacity=0][fill={rgb, 255:red, 74; green, 144; blue, 226 }  ,fill opacity=0.21 ] (309.58,412.75) .. controls (313.67,406.93) and (325.08,415) .. (342.33,413.75) .. controls (359.58,412.5) and (361.83,406.25) .. (373.08,411.25) .. controls (384.33,416.25) and (369.51,423.03) .. (357.42,425.14) .. controls (345.33,427.25) and (331.58,427.5) .. (321.92,425.89) .. controls (312.26,424.28) and (305.5,418.57) .. (309.58,412.75) -- cycle ;
\draw    (342.1,399.18) -- (342.6,420.18) ;
\draw [shift={(342.38,411.08)}, rotate = 268.64] [fill={rgb, 255:red, 0; green, 0; blue, 0 }  ][line width=0.08]  [draw opacity=0] (5.36,-2.57) -- (0,0) -- (5.36,2.57) -- (3.56,0) -- cycle    ;
\draw    (315.86,392.97) -- (315.89,413.54) ;
\draw [shift={(315.88,404.66)}, rotate = 269.92] [fill={rgb, 255:red, 0; green, 0; blue, 0 }  ][line width=0.08]  [draw opacity=0] (5.36,-2.57) -- (0,0) -- (5.36,2.57) -- (3.56,0) -- cycle    ;
\draw    (367.34,393.8) -- (367.37,414.37) ;
\draw [shift={(367.35,405.48)}, rotate = 269.92] [fill={rgb, 255:red, 0; green, 0; blue, 0 }  ][line width=0.08]  [draw opacity=0] (5.36,-2.57) -- (0,0) -- (5.36,2.57) -- (3.56,0) -- cycle    ;
\draw    (341.76,375.23) -- (367.34,393.8) ;
\draw [shift={(355.68,385.34)}, rotate = 215.96] [fill={rgb, 255:red, 0; green, 0; blue, 0 }  ][line width=0.08]  [draw opacity=0] (5.36,-2.57) -- (0,0) -- (5.36,2.57) -- (3.56,0) -- cycle    ;
\draw  [fill={rgb, 255:red, 6; green, 80; blue, 250 }  ,fill opacity=1 ] (364.7,393.8) .. controls (364.7,392.35) and (365.88,391.18) .. (367.34,391.18) .. controls (368.79,391.18) and (369.97,392.35) .. (369.97,393.8) .. controls (369.97,395.24) and (368.79,396.41) .. (367.34,396.41) .. controls (365.88,396.41) and (364.7,395.24) .. (364.7,393.8) -- cycle ;
\draw [color={rgb, 255:red, 161; green, 27; blue, 43 }  ,draw opacity=1 ]   (341.69,351.42) .. controls (343.35,353.09) and (343.35,354.75) .. (341.68,356.42) -- (341.67,357.89) -- (341.66,365.89) ;
\draw [shift={(341.66,368.89)}, rotate = 270.1] [fill={rgb, 255:red, 161; green, 27; blue, 43 }  ,fill opacity=1 ][line width=0.08]  [draw opacity=0] (7.14,-3.43) -- (0,0) -- (7.14,3.43) -- (4.74,0) -- cycle    ;
\draw  [fill={rgb, 255:red, 239; green, 43; blue, 67 }  ,fill opacity=1 ] (339.05,348.8) .. controls (339.05,347.35) and (340.23,346.18) .. (341.69,346.18) .. controls (343.14,346.18) and (344.32,347.35) .. (344.32,348.8) .. controls (344.32,350.24) and (343.14,351.42) .. (341.69,351.42) .. controls (340.23,351.42) and (339.05,350.24) .. (339.05,348.8) -- cycle ;
\draw    (341.76,375.23) -- (315.86,392.97) ;
\draw [shift={(327.66,384.89)}, rotate = 325.59] [fill={rgb, 255:red, 0; green, 0; blue, 0 }  ][line width=0.08]  [draw opacity=0] (5.36,-2.57) -- (0,0) -- (5.36,2.57) -- (3.56,0) -- cycle    ;
\draw    (341.77,375.4) -- (342.1,399.18) ;
\draw [shift={(341.95,388.69)}, rotate = 269.22] [fill={rgb, 255:red, 0; green, 0; blue, 0 }  ][line width=0.08]  [draw opacity=0] (5.36,-2.57) -- (0,0) -- (5.36,2.57) -- (3.56,0) -- cycle    ;
\draw  [color={rgb, 255:red, 64; green, 95; blue, 17 }  ,draw opacity=1 ][fill={rgb, 255:red, 85; green, 211; blue, 33 }  ,fill opacity=1 ] (341.66,369.7) .. controls (343.13,369.67) and (344.35,370.92) .. (344.38,372.49) .. controls (344.41,374.07) and (343.25,375.37) .. (341.77,375.4) .. controls (340.3,375.43) and (339.08,374.18) .. (339.05,372.61) .. controls (339.02,371.03) and (340.18,369.73) .. (341.66,369.7) -- cycle ;
\draw  [fill={rgb, 255:red, 6; green, 80; blue, 250 }  ,fill opacity=1 ] (313.23,392.97) .. controls (313.23,391.52) and (314.41,390.35) .. (315.86,390.35) .. controls (317.32,390.35) and (318.5,391.52) .. (318.5,392.97) .. controls (318.5,394.41) and (317.32,395.59) .. (315.86,395.59) .. controls (314.41,395.59) and (313.23,394.41) .. (313.23,392.97) -- cycle ;
\draw  [fill={rgb, 255:red, 6; green, 80; blue, 250 }  ,fill opacity=1 ] (339.46,399.18) .. controls (339.46,397.74) and (340.64,396.57) .. (342.1,396.57) .. controls (343.55,396.57) and (344.73,397.74) .. (344.73,399.18) .. controls (344.73,400.63) and (343.55,401.8) .. (342.1,401.8) .. controls (340.64,401.8) and (339.46,400.63) .. (339.46,399.18) -- cycle ;
\draw  [fill={rgb, 255:red, 6; green, 80; blue, 250 }  ,fill opacity=1 ] (313.26,413.54) .. controls (313.26,412.1) and (314.44,410.93) .. (315.89,410.93) .. controls (317.35,410.93) and (318.53,412.1) .. (318.53,413.54) .. controls (318.53,414.99) and (317.35,416.16) .. (315.89,416.16) .. controls (314.44,416.16) and (313.26,414.99) .. (313.26,413.54) -- cycle ;
\draw  [fill={rgb, 255:red, 6; green, 80; blue, 250 }  ,fill opacity=1 ] (365.26,413.54) .. controls (365.26,412.1) and (366.44,410.93) .. (367.89,410.93) .. controls (369.35,410.93) and (370.53,412.1) .. (370.53,413.54) .. controls (370.53,414.99) and (369.35,416.16) .. (367.89,416.16) .. controls (366.44,416.16) and (365.26,414.99) .. (365.26,413.54) -- cycle ;
\draw  [fill={rgb, 255:red, 6; green, 80; blue, 250 }  ,fill opacity=1 ] (339.96,420.18) .. controls (339.96,418.74) and (341.14,417.57) .. (342.6,417.57) .. controls (344.05,417.57) and (345.23,418.74) .. (345.23,420.18) .. controls (345.23,421.63) and (344.05,422.8) .. (342.6,422.8) .. controls (341.14,422.8) and (339.96,421.63) .. (339.96,420.18) -- cycle ;
\draw  [draw opacity=0][fill={rgb, 255:red, 74; green, 144; blue, 226 }  ,fill opacity=0.19 ] (318.71,372.55) .. controls (318.71,368.5) and (329.01,365.22) .. (341.71,365.22) .. controls (354.42,365.22) and (364.71,368.5) .. (364.71,372.55) .. controls (364.71,376.6) and (354.42,379.88) .. (341.71,379.88) .. controls (329.01,379.88) and (318.71,376.6) .. (318.71,372.55) -- cycle ;

\draw (346.8,365.69) node [anchor=north west][inner sep=0.75pt]  [font=\tiny]  {$1$};
\draw (304.54,387.16) node [anchor=north west][inner sep=0.75pt]  [font=\tiny]  {$2$};
\draw (345.59,396.86) node [anchor=north west][inner sep=0.75pt]  [font=\tiny]  {$3$};
\draw (370.76,390.28) node [anchor=north west][inner sep=0.75pt]  [font=\tiny]  {$4$};
\draw (337.92,337.19) node [anchor=north west][inner sep=0.75pt]  [font=\tiny]  {$u(t)$};
\draw (345.78,354.22) node [anchor=north west][inner sep=0.75pt]  [font=\tiny,color={rgb, 255:red, 128; green, 128; blue, 128 }  ,opacity=1 ]  {$+$};
\draw (330.81,384.5) node [anchor=north west][inner sep=0.75pt]  [font=\tiny,color={rgb, 255:red, 128; green, 128; blue, 128 }  ,opacity=1 ]  {$-$};
\draw (355.81,376.5) node [anchor=north west][inner sep=0.75pt]  [font=\tiny,color={rgb, 255:red, 128; green, 128; blue, 128 }  ,opacity=1 ]  {$-$};
\draw (319.78,374.72) node [anchor=north west][inner sep=0.75pt]  [font=\tiny,color={rgb, 255:red, 128; green, 128; blue, 128 }  ,opacity=1 ]  {$+$};
\draw (303.04,412.16) node [anchor=north west][inner sep=0.75pt]  [font=\tiny]  {$5$};
\draw (347.54,416.66) node [anchor=north west][inner sep=0.75pt]  [font=\tiny]  {$6$};
\draw (373.74,410.16) node [anchor=north west][inner sep=0.75pt]  [font=\tiny]  {$7$};
\draw (330.78,408.22) node [anchor=north west][inner sep=0.75pt]  [font=\tiny,color={rgb, 255:red, 128; green, 128; blue, 128 }  ,opacity=1 ]  {$+$};
\draw (373.31,400.5) node [anchor=north west][inner sep=0.75pt]  [font=\tiny,color={rgb, 255:red, 128; green, 128; blue, 128 }  ,opacity=1 ]  {$-$};
\draw (304.28,401.22) node [anchor=north west][inner sep=0.75pt]  [font=\tiny,color={rgb, 255:red, 128; green, 128; blue, 128 }  ,opacity=1 ]  {$+$};
\draw (322.5,394.23) node [anchor=north west][inner sep=0.75pt]  [font=\tiny,color={rgb, 255:red, 209; green, 56; blue, 15 }  ,opacity=1 ]  {$L_{2}$};
\draw (320.53,416.94) node [anchor=north west][inner sep=0.75pt]  [font=\tiny,color={rgb, 255:red, 209; green, 56; blue, 15 }  ,opacity=1 ]  {$L_{3}$};
\draw (326,367.73) node [anchor=north west][inner sep=0.75pt]  [font=\tiny,color={rgb, 255:red, 209; green, 56; blue, 15 }  ,opacity=1 ]  {$L_{1}$};

\end{tikzpicture}
\caption{A digraph $\mathcal{G}_2$ $(\mathcal{A,B})$}
\label{SM1}
\end{subfigure}
 \begin{subfigure}{0.14\textwidth} 
    \centering

\tikzset{every picture/.style={line width=0.75pt}} 

\begin{tikzpicture}[x=0.75pt,y=0.75pt,yscale=-1,xscale=1]

\draw  [draw opacity=0][fill={rgb, 255:red, 74; green, 144; blue, 226 }  ,fill opacity=0.21 ] (100.83,532) .. controls (104.92,526.18) and (116.33,534.25) .. (133.58,533) .. controls (150.83,531.75) and (153.08,525.5) .. (164.33,530.5) .. controls (175.58,535.5) and (160.76,542.28) .. (148.67,544.39) .. controls (136.58,546.5) and (122.83,546.75) .. (113.17,545.14) .. controls (103.51,543.53) and (96.75,537.82) .. (100.83,532) -- cycle ;
\draw  [draw opacity=0][fill={rgb, 255:red, 74; green, 144; blue, 226 }  ,fill opacity=0.21 ] (101.58,553.75) .. controls (105.67,547.93) and (117.08,556) .. (134.33,554.75) .. controls (151.58,553.5) and (153.83,547.25) .. (165.08,552.25) .. controls (176.33,557.25) and (161.51,564.03) .. (149.42,566.14) .. controls (137.33,568.25) and (123.58,568.5) .. (113.92,566.89) .. controls (104.26,565.28) and (97.5,559.57) .. (101.58,553.75) -- cycle ;
\draw [color={rgb, 255:red, 128; green, 128; blue, 128 }  ,draw opacity=0.35 ] [dash pattern={on 1.5pt off 0.75pt on 0.75pt off 1.5pt}]  (133.64,516.39) -- (119.26,526.14) -- (107.73,533.96) ;
\draw [shift={(125.29,522.05)}, rotate = 325.86] [fill={rgb, 255:red, 128; green, 128; blue, 128 }  ,fill opacity=0.35 ][line width=0.08]  [draw opacity=0] (5.36,-2.57) -- (0,0) -- (5.36,2.57) -- (3.56,0) -- cycle    ;
\draw [shift={(112.34,530.84)}, rotate = 325.86] [fill={rgb, 255:red, 128; green, 128; blue, 128 }  ,fill opacity=0.35 ][line width=0.08]  [draw opacity=0] (5.36,-2.57) -- (0,0) -- (5.36,2.57) -- (3.56,0) -- cycle    ;
\draw [color={rgb, 255:red, 128; green, 128; blue, 128 }  ,draw opacity=0.35 ] [dash pattern={on 1.5pt off 0.75pt on 0.75pt off 1.5pt}]  (134.16,542.82) -- (134.66,563.82) ;
\draw [shift={(134.45,554.72)}, rotate = 268.64] [fill={rgb, 255:red, 128; green, 128; blue, 128 }  ,fill opacity=0.35 ][line width=0.08]  [draw opacity=0] (5.36,-2.57) -- (0,0) -- (5.36,2.57) -- (3.56,0) -- cycle    ;
\draw [color={rgb, 255:red, 154; green, 36; blue, 50 }  ,draw opacity=1 ]   (107.93,536.6) -- (107.96,557.18) ;
\draw [shift={(107.95,548.29)}, rotate = 269.92] [fill={rgb, 255:red, 154; green, 36; blue, 50 }  ,fill opacity=1 ][line width=0.08]  [draw opacity=0] (5.36,-2.57) -- (0,0) -- (5.36,2.57) -- (3.56,0) -- cycle    ;
\draw [color={rgb, 255:red, 154; green, 36; blue, 50 }  ,draw opacity=1 ]   (159.4,537.43) -- (159.43,558) ;
\draw [shift={(159.42,549.12)}, rotate = 269.92] [fill={rgb, 255:red, 154; green, 36; blue, 50 }  ,fill opacity=1 ][line width=0.08]  [draw opacity=0] (5.36,-2.57) -- (0,0) -- (5.36,2.57) -- (3.56,0) -- cycle    ;
\draw [color={rgb, 255:red, 154; green, 36; blue, 50 }  ,draw opacity=1 ]   (133.67,515.33) -- (159.25,533.89) ;
\draw [shift={(147.59,525.44)}, rotate = 215.96] [fill={rgb, 255:red, 154; green, 36; blue, 50 }  ,fill opacity=1 ][line width=0.08]  [draw opacity=0] (5.36,-2.57) -- (0,0) -- (5.36,2.57) -- (3.56,0) -- cycle    ;
\draw [color={rgb, 255:red, 154; green, 36; blue, 50 }  ,draw opacity=1 ]   (133.67,515.33) -- (133.99,539.12) ;
\draw [shift={(133.85,528.62)}, rotate = 269.22] [fill={rgb, 255:red, 154; green, 36; blue, 50 }  ,fill opacity=1 ][line width=0.08]  [draw opacity=0] (5.36,-2.57) -- (0,0) -- (5.36,2.57) -- (3.56,0) -- cycle    ;
\draw  [fill={rgb, 255:red, 6; green, 80; blue, 250 }  ,fill opacity=1 ] (156.57,534.79) .. controls (156.57,533.34) and (157.75,532.17) .. (159.2,532.17) .. controls (160.66,532.17) and (161.84,533.34) .. (161.84,534.79) .. controls (161.84,536.23) and (160.66,537.41) .. (159.2,537.41) .. controls (157.75,537.41) and (156.57,536.23) .. (156.57,534.79) -- cycle ;
\draw [color={rgb, 255:red, 161; green, 27; blue, 43 }  ,draw opacity=1 ]   (133.55,492.41) .. controls (135.22,494.08) and (135.21,495.75) .. (133.54,497.41) -- (133.54,498.89) -- (133.53,506.89) ;
\draw [shift={(133.52,509.89)}, rotate = 270.1] [fill={rgb, 255:red, 161; green, 27; blue, 43 }  ,fill opacity=1 ][line width=0.08]  [draw opacity=0] (7.14,-3.43) -- (0,0) -- (7.14,3.43) -- (4.74,0) -- cycle    ;
\draw  [fill={rgb, 255:red, 239; green, 43; blue, 67 }  ,fill opacity=1 ] (130.92,489.79) .. controls (130.92,488.35) and (132.1,487.17) .. (133.55,487.17) .. controls (135.01,487.17) and (136.18,488.35) .. (136.18,489.79) .. controls (136.18,491.24) and (135.01,492.41) .. (133.55,492.41) .. controls (132.1,492.41) and (130.92,491.24) .. (130.92,489.79) -- cycle ;
\draw  [color={rgb, 255:red, 64; green, 95; blue, 17 }  ,draw opacity=1 ][fill={rgb, 255:red, 85; green, 211; blue, 33 }  ,fill opacity=1 ] (133.52,510.69) .. controls (134.99,510.66) and (136.21,511.91) .. (136.25,513.49) .. controls (136.28,515.06) and (135.11,516.36) .. (133.64,516.39) .. controls (132.17,516.42) and (130.95,515.17) .. (130.91,513.6) .. controls (130.88,512.03) and (132.05,510.73) .. (133.52,510.69) -- cycle ;
\draw  [fill={rgb, 255:red, 6; green, 80; blue, 250 }  ,fill opacity=1 ] (105.1,533.96) .. controls (105.1,532.52) and (106.28,531.35) .. (107.73,531.35) .. controls (109.18,531.35) and (110.36,532.52) .. (110.36,533.96) .. controls (110.36,535.41) and (109.18,536.58) .. (107.73,536.58) .. controls (106.28,536.58) and (105.1,535.41) .. (105.1,533.96) -- cycle ;
\draw  [fill={rgb, 255:red, 6; green, 80; blue, 250 }  ,fill opacity=1 ] (131.33,540.18) .. controls (131.33,538.73) and (132.51,537.56) .. (133.96,537.56) .. controls (135.42,537.56) and (136.59,538.73) .. (136.59,540.18) .. controls (136.59,541.62) and (135.42,542.79) .. (133.96,542.79) .. controls (132.51,542.79) and (131.33,541.62) .. (131.33,540.18) -- cycle ;
\draw  [fill={rgb, 255:red, 6; green, 80; blue, 250 }  ,fill opacity=1 ] (105.33,557.18) .. controls (105.33,555.73) and (106.51,554.56) .. (107.96,554.56) .. controls (109.42,554.56) and (110.59,555.73) .. (110.59,557.18) .. controls (110.59,558.62) and (109.42,559.79) .. (107.96,559.79) .. controls (106.51,559.79) and (105.33,558.62) .. (105.33,557.18) -- cycle ;
\draw  [fill={rgb, 255:red, 6; green, 80; blue, 250 }  ,fill opacity=1 ] (156.33,557.18) .. controls (156.33,555.73) and (157.51,554.56) .. (158.96,554.56) .. controls (160.42,554.56) and (161.59,555.73) .. (161.59,557.18) .. controls (161.59,558.62) and (160.42,559.79) .. (158.96,559.79) .. controls (157.51,559.79) and (156.33,558.62) .. (156.33,557.18) -- cycle ;
\draw  [fill={rgb, 255:red, 6; green, 80; blue, 250 }  ,fill opacity=1 ] (132.03,563.82) .. controls (132.03,562.37) and (133.21,561.2) .. (134.66,561.2) .. controls (136.12,561.2) and (137.3,562.37) .. (137.3,563.82) .. controls (137.3,565.26) and (136.12,566.43) .. (134.66,566.43) .. controls (133.21,566.43) and (132.03,565.26) .. (132.03,563.82) -- cycle ;
\draw  [draw opacity=0][fill={rgb, 255:red, 74; green, 144; blue, 226 }  ,fill opacity=0.19 ] (110.58,513.54) .. controls (110.58,509.49) and (120.88,506.21) .. (133.58,506.21) .. controls (146.28,506.21) and (156.58,509.49) .. (156.58,513.54) .. controls (156.58,517.59) and (146.28,520.88) .. (133.58,520.88) .. controls (120.88,520.88) and (110.58,517.59) .. (110.58,513.54) -- cycle ;

\draw (96.4,528.15) node [anchor=north west][inner sep=0.75pt]  [font=\tiny]  {$2$};
\draw (137.45,537.85) node [anchor=north west][inner sep=0.75pt]  [font=\tiny]  {$3$};
\draw (162.63,531.27) node [anchor=north west][inner sep=0.75pt]  [font=\tiny]  {$4$};
\draw (139.85,495.36) node [anchor=north west][inner sep=0.75pt]  [font=\tiny,color={rgb, 255:red, 128; green, 128; blue, 128 }  ,opacity=1 ]  {$+$};
\draw (113.85,515.86) node [anchor=north west][inner sep=0.75pt]  [font=\tiny,color={rgb, 255:red, 128; green, 128; blue, 128 }  ,opacity=1 ]  {$+$};
\draw (122.7,527.44) node [anchor=north west][inner sep=0.75pt]  [font=\tiny,color={rgb, 255:red, 128; green, 128; blue, 128 }  ,opacity=1 ]  {$-$};
\draw (149.38,515.14) node [anchor=north west][inner sep=0.75pt]  [font=\tiny,color={rgb, 255:red, 128; green, 128; blue, 128 }  ,opacity=1 ]  {$-$};
\draw (93.6,553.29) node [anchor=north west][inner sep=0.75pt]  [font=\tiny]  {$5$};
\draw (138.1,557.79) node [anchor=north west][inner sep=0.75pt]  [font=\tiny]  {$6$};
\draw (165.1,555.29) node [anchor=north west][inner sep=0.75pt]  [font=\tiny]  {$7$};
\draw (121.35,549.36) node [anchor=north west][inner sep=0.75pt]  [font=\tiny,color={rgb, 255:red, 128; green, 128; blue, 128 }  ,opacity=1 ]  {$+$};
\draw (163.88,541.64) node [anchor=north west][inner sep=0.75pt]  [font=\tiny,color={rgb, 255:red, 128; green, 128; blue, 128 }  ,opacity=1 ]  {$-$};
\draw (94.85,542.36) node [anchor=north west][inner sep=0.75pt]  [font=\tiny,color={rgb, 255:red, 128; green, 128; blue, 128 }  ,opacity=1 ]  {$+$};
\draw (127.78,475.47) node [anchor=north west][inner sep=0.75pt]  [font=\tiny]  {$u( t)$};
\draw (140.8,509.69) node [anchor=north west][inner sep=0.75pt]  [font=\tiny]  {$1$};

\end{tikzpicture}
\caption{Snapshot 1}
\label{sm2}
\end{subfigure}
   \begin{subfigure}{0.14\textwidth} 
    \centering

\tikzset{every picture/.style={line width=0.75pt}} 

\begin{tikzpicture}[x=0.75pt,y=0.75pt,yscale=-1,xscale=1]

\draw  [draw opacity=0][fill={rgb, 255:red, 74; green, 144; blue, 226 }  ,fill opacity=0.19 ] (532.72,316.44) .. controls (532.72,312.39) and (543.01,309.11) .. (555.72,309.11) .. controls (568.42,309.11) and (578.72,312.39) .. (578.72,316.44) .. controls (578.72,320.49) and (568.42,323.78) .. (555.72,323.78) .. controls (543.01,323.78) and (532.72,320.49) .. (532.72,316.44) -- cycle ;
\draw  [draw opacity=0][fill={rgb, 255:red, 74; green, 144; blue, 226 }  ,fill opacity=0.21 ] (525.83,333.67) .. controls (529.92,327.85) and (541.33,335.92) .. (558.58,334.67) .. controls (575.83,333.42) and (578.08,327.17) .. (589.33,332.17) .. controls (600.58,337.17) and (585.76,343.94) .. (573.67,346.05) .. controls (561.58,348.17) and (547.83,348.42) .. (538.17,346.8) .. controls (528.51,345.19) and (521.75,339.49) .. (525.83,333.67) -- cycle ;
\draw  [draw opacity=0][fill={rgb, 255:red, 74; green, 144; blue, 226 }  ,fill opacity=0.21 ] (526.58,355.42) .. controls (530.67,349.6) and (542.08,357.67) .. (559.33,356.42) .. controls (576.58,355.17) and (578.83,348.92) .. (590.08,353.92) .. controls (601.33,358.92) and (586.51,365.69) .. (574.42,367.8) .. controls (562.33,369.92) and (548.58,370.17) .. (538.92,368.55) .. controls (529.26,366.94) and (522.5,361.24) .. (526.58,355.42) -- cycle ;
\draw [color={rgb, 255:red, 135; green, 38; blue, 52 }  ,draw opacity=1 ]   (556.23,344.52) -- (556.73,365.52) ;
\draw [shift={(556.51,356.42)}, rotate = 268.64] [fill={rgb, 255:red, 135; green, 38; blue, 52 }  ,fill opacity=1 ][line width=0.08]  [draw opacity=0] (5.36,-2.57) -- (0,0) -- (5.36,2.57) -- (3.56,0) -- cycle    ;
\draw [color={rgb, 255:red, 152; green, 22; blue, 36 }  ,draw opacity=0.24 ]   (530,338.3) -- (530.03,358.88) ;
\draw [shift={(530.02,349.99)}, rotate = 269.92] [fill={rgb, 255:red, 152; green, 22; blue, 36 }  ,fill opacity=0.24 ][line width=0.08]  [draw opacity=0] (5.36,-2.57) -- (0,0) -- (5.36,2.57) -- (3.56,0) -- cycle    ;
\draw [color={rgb, 255:red, 152; green, 22; blue, 36 }  ,draw opacity=0.23 ]   (581.47,339.13) -- (581.5,359.7) ;
\draw [shift={(581.49,350.82)}, rotate = 269.92] [fill={rgb, 255:red, 152; green, 22; blue, 36 }  ,fill opacity=0.23 ][line width=0.08]  [draw opacity=0] (5.36,-2.57) -- (0,0) -- (5.36,2.57) -- (3.56,0) -- cycle    ;
\draw [color={rgb, 255:red, 135; green, 38; blue, 52 }  ,draw opacity=1 ]   (555.72,319.37) -- (529.83,337.1) ;
\draw [shift={(541.62,329.03)}, rotate = 325.59] [fill={rgb, 255:red, 135; green, 38; blue, 52 }  ,fill opacity=1 ][line width=0.08]  [draw opacity=0] (5.36,-2.57) -- (0,0) -- (5.36,2.57) -- (3.56,0) -- cycle    ;
\draw [color={rgb, 255:red, 152; green, 22; blue, 36 }  ,draw opacity=0.24 ]   (555.72,318.87) -- (581.3,337.43) ;
\draw [shift={(569.64,328.97)}, rotate = 215.96] [fill={rgb, 255:red, 152; green, 22; blue, 36 }  ,fill opacity=0.24 ][line width=0.08]  [draw opacity=0] (5.36,-2.57) -- (0,0) -- (5.36,2.57) -- (3.56,0) -- cycle    ;
\draw [color={rgb, 255:red, 152; green, 22; blue, 36 }  ,draw opacity=0.24 ]   (555.73,319.03) -- (556.06,342.82) ;
\draw [shift={(555.91,332.32)}, rotate = 269.22] [fill={rgb, 255:red, 152; green, 22; blue, 36 }  ,fill opacity=0.24 ][line width=0.08]  [draw opacity=0] (5.36,-2.57) -- (0,0) -- (5.36,2.57) -- (3.56,0) -- cycle    ;
\draw  [fill={rgb, 255:red, 6; green, 80; blue, 250 }  ,fill opacity=1 ] (578.66,337.43) .. controls (578.66,335.98) and (579.84,334.81) .. (581.3,334.81) .. controls (582.75,334.81) and (583.93,335.98) .. (583.93,337.43) .. controls (583.93,338.87) and (582.75,340.05) .. (581.3,340.05) .. controls (579.84,340.05) and (578.66,338.87) .. (578.66,337.43) -- cycle ;
\draw [color={rgb, 255:red, 161; green, 27; blue, 43 }  ,draw opacity=1 ]   (555.65,295.05) .. controls (557.31,296.72) and (557.31,298.38) .. (555.64,300.05) -- (555.64,301.53) -- (555.62,309.53) ;
\draw [shift={(555.62,312.53)}, rotate = 270.1] [fill={rgb, 255:red, 161; green, 27; blue, 43 }  ,fill opacity=1 ][line width=0.08]  [draw opacity=0] (7.14,-3.43) -- (0,0) -- (7.14,3.43) -- (4.74,0) -- cycle    ;
\draw  [fill={rgb, 255:red, 239; green, 43; blue, 67 }  ,fill opacity=1 ] (553.01,292.43) .. controls (553.01,290.99) and (554.19,289.81) .. (555.65,289.81) .. controls (557.1,289.81) and (558.28,290.99) .. (558.28,292.43) .. controls (558.28,293.88) and (557.1,295.05) .. (555.65,295.05) .. controls (554.19,295.05) and (553.01,293.88) .. (553.01,292.43) -- cycle ;
\draw  [color={rgb, 255:red, 64; green, 95; blue, 17 }  ,draw opacity=1 ][fill={rgb, 255:red, 85; green, 211; blue, 33 }  ,fill opacity=1 ] (555.62,313.33) .. controls (557.09,313.3) and (558.31,314.55) .. (558.34,316.13) .. controls (558.37,317.7) and (557.21,319) .. (555.73,319.03) .. controls (554.26,319.06) and (553.04,317.81) .. (553.01,316.24) .. controls (552.98,314.67) and (554.14,313.37) .. (555.62,313.33) -- cycle ;
\draw  [fill={rgb, 255:red, 6; green, 80; blue, 250 }  ,fill opacity=1 ] (527.19,336.6) .. controls (527.19,335.16) and (528.37,333.99) .. (529.83,333.99) .. controls (531.28,333.99) and (532.46,335.16) .. (532.46,336.6) .. controls (532.46,338.05) and (531.28,339.22) .. (529.83,339.22) .. controls (528.37,339.22) and (527.19,338.05) .. (527.19,336.6) -- cycle ;
\draw  [fill={rgb, 255:red, 6; green, 80; blue, 250 }  ,fill opacity=1 ] (553.42,342.82) .. controls (553.42,341.37) and (554.6,340.2) .. (556.06,340.2) .. controls (557.51,340.2) and (558.69,341.37) .. (558.69,342.82) .. controls (558.69,344.26) and (557.51,345.43) .. (556.06,345.43) .. controls (554.6,345.43) and (553.42,344.26) .. (553.42,342.82) -- cycle ;
\draw  [fill={rgb, 255:red, 6; green, 80; blue, 250 }  ,fill opacity=1 ] (527.4,358.88) .. controls (527.4,357.43) and (528.57,356.26) .. (530.03,356.26) .. controls (531.48,356.26) and (532.66,357.43) .. (532.66,358.88) .. controls (532.66,360.32) and (531.48,361.49) .. (530.03,361.49) .. controls (528.57,361.49) and (527.4,360.32) .. (527.4,358.88) -- cycle ;
\draw  [fill={rgb, 255:red, 6; green, 80; blue, 250 }  ,fill opacity=1 ] (578.4,358.88) .. controls (578.4,357.43) and (579.57,356.26) .. (581.03,356.26) .. controls (582.48,356.26) and (583.66,357.43) .. (583.66,358.88) .. controls (583.66,360.32) and (582.48,361.49) .. (581.03,361.49) .. controls (579.57,361.49) and (578.4,360.32) .. (578.4,358.88) -- cycle ;
\draw  [fill={rgb, 255:red, 6; green, 80; blue, 250 }  ,fill opacity=1 ] (554.1,365.52) .. controls (554.1,364.07) and (555.27,362.9) .. (556.73,362.9) .. controls (558.18,362.9) and (559.36,364.07) .. (559.36,365.52) .. controls (559.36,366.96) and (558.18,368.13) .. (556.73,368.13) .. controls (555.27,368.13) and (554.1,366.96) .. (554.1,365.52) -- cycle ;

\draw (518.5,330.79) node [anchor=north west][inner sep=0.75pt]  [font=\tiny]  {$2$};
\draw (559.55,340.49) node [anchor=north west][inner sep=0.75pt]  [font=\tiny]  {$3$};
\draw (584.72,333.91) node [anchor=north west][inner sep=0.75pt]  [font=\tiny]  {$4$};
\draw (561.42,298.56) node [anchor=north west][inner sep=0.75pt]  [font=\tiny,color={rgb, 255:red, 128; green, 128; blue, 128 }  ,opacity=1 ]  {$+$};
\draw (517.67,355.99) node [anchor=north west][inner sep=0.75pt]  [font=\tiny]  {$5$};
\draw (562.17,360.49) node [anchor=north west][inner sep=0.75pt]  [font=\tiny]  {$6$};
\draw (586.17,356.99) node [anchor=north west][inner sep=0.75pt]  [font=\tiny]  {$7$};
\draw (545.42,352.06) node [anchor=north west][inner sep=0.75pt]  [font=\tiny,color={rgb, 255:red, 128; green, 128; blue, 128 }  ,opacity=1 ]  {$+$};
\draw (587.94,344.34) node [anchor=north west][inner sep=0.75pt]  [font=\tiny,color={rgb, 255:red, 128; green, 128; blue, 128 }  ,opacity=1 ]  {$-$};
\draw (518.92,345.06) node [anchor=north west][inner sep=0.75pt]  [font=\tiny,color={rgb, 255:red, 128; green, 128; blue, 128 }  ,opacity=1 ]  {$+$};
\draw (549.04,279.27) node [anchor=north west][inner sep=0.75pt]  [font=\tiny]  {$u( t)$};
\draw (562.8,312.19) node [anchor=north west][inner sep=0.75pt]  [font=\tiny]  {$1$};
\draw (532.85,321.86) node [anchor=north west][inner sep=0.75pt]  [font=\tiny,color={rgb, 255:red, 128; green, 128; blue, 128 }  ,opacity=1 ]  {$+$};
\draw (546.7,328.44) node [anchor=north west][inner sep=0.75pt]  [font=\tiny,color={rgb, 255:red, 128; green, 128; blue, 128 }  ,opacity=1 ]  {$-$};
\draw (571.38,319.14) node [anchor=north west][inner sep=0.75pt]  [font=\tiny,color={rgb, 255:red, 128; green, 128; blue, 128 }  ,opacity=1 ]  {$-$};
\end{tikzpicture}
\caption{Snapshot 2}
\label{sm3}
\end{subfigure}
\caption{The sign of the sum of the products associated with nodes $\{5,7\}$ from the leader node is identical in Snapshot~1. The nodes $\{1,3,4,5,7\}$ are path-sign-matched in Snapshot~1, whereas nodes $\{2,6\}$ are path-sign-matched in Snapshot~2.}\label{total SM}
\end{figure}

\begin{lemma} \label{first snap}
The number of herdable nodes in the first snapshot equals the number of path-sign-matched nodes in the static version of the digraph.
\end{lemma}

    


\begin{theorem}\label{min snap}
Consider the digraph \(\mathcal{G}_T(\mathcal{A}_i,\mathcal{B})\) associated with the linear temporally switching network described by \eqref{sys1}. Let \(\mathcal{G}_{s(i)}\) be the signed layered graph associated with \((\mathcal{A}_i,\mathcal{B})\) on the interval \([t_{i-1},t_i)\). Suppose that the system in \eqref{sys1} is not \(\mathcal{SS}\) herdable during the first snapshot (i.e., in its static form). Then, the system becomes \(\mathcal{SS}\) herdable if and only if every node that is not path-sign-matched in snapshot \(i\) becomes path-sign-matched in a subsequent snapshot according to Definition~\ref{path sign}.
\end{theorem}
 
\begin{proof} Sufficiency :Let us consider a digraph that is not $\mathcal{SS}$ herdable in its $i^{th}$ snapshot (static form). This loss of herdability arises due to the presence of signed dilation and layered dilation in the signed layered representation of the digraph. These dilations imply that the nodes belonging to the dilation sets produce rows in the controllability matrix that are non–negatively linearly dependent, as illustrated below.
\renewcommand{\arraystretch}{1.2}
{\small
\[
\mathcal{C}_{(i,j)} =
\begin{bNiceMatrix}
\mathcal{C}_{11} & 0 & \cdots & 0 & \cdots & 0 \\
. & . & . & . & . & . \\
\rowcolor{gray!15}
0 & \mathcal{C}_{p2} & \cdots & \mathcal{C}_{pi} & \cdots & \mathcal{C}_{pn} \\
.& . & . & . & . & . \\
\rowcolor{gray!15}
0 & -\alpha \mathcal{C}_{p2} & \cdots & -\alpha \mathcal{C}_{pi} & \cdots & -\alpha \mathcal{C}_{pn} \\
. & . & \ddots & \ddots & . & . \\
0 & \mathcal{C}_{n2} & \cdots & \mathcal{C}_{ni} & \cdots & \mathcal{C}_{nn}
\end{bNiceMatrix},
\qquad \alpha > 0; \alpha \in \mathbb{R}
\]
}

Without loss of generality, suppose that the first \(q\) nodes are path-sign-matched in the \(i^{\text{th}}\) snapshot. By Lemma~\ref{first snap}, these \(q\) nodes are herdable in the \(i^{\text{th}}\) snapshot, whereas the remaining \(n-q\) nodes are not herdable. Consequently, at least \(n-q\) rows of the controllability matrix are nonnegatively linearly dependent on the rows corresponding to the herdable nodes. 

Now suppose that the remaining \(n-q\) nodes become path-sign-matched in subsequent snapshots. Then, by Definition~\ref{path sign}, the corresponding entries in the columns associated with those snapshots attain different magnitudes under a different parametric realization while preserving the sign pattern. As a result, the corresponding rows of the temporal controllability matrix \(\mathcal{C}_T\) are no longer non-negatively linearly dependent. Therefore, by Proposition~\ref{test for H}, the temporally switching network with fixed topology and structurally similar system matrices is \(\mathcal{SS}\) herdable.

Necessity: Consider a temporally switching digraph containing a signed dilation or layer dilation that is \(\mathcal{SS}\) herdable. From Lemma~\ref{SLG} and Sections~\ref{SD para} and \ref{Lyr para}, the entries corresponding to the signed and layer dilations induce non-negatively linearly dependent rows in the controllability matrix. Consequently, during the first snapshot, only a subset of the nodes in the layer associated with these dilations are herdable, whereas the nodes corresponding to such dependencies are not herdable.

By Proposition~\ref{test for H}, a digraph is herdable only if no two rows of the controllability matrix associated with the temporally switching network are non-negatively linearly dependent. Furthermore, as shown in Section~\ref{Sec:3}, a temporally switching digraph with identical edge weights across all snapshots does not improve herdability. Therefore, in order to eliminate the nonnegative linear dependence among the rows, there must exist a different parametric realization of \((\mathcal{A}_i,\mathcal{B})\) such that the unmatched nodes in the \(i^{\text{th}}\) snapshot become path-sign-matched in subsequent snapshots. This results in distinct realizations in the snapshots following the \(i^{\text{th}}\) snapshot, with the system evolution governed by \(\prod_{k=i+1}^{p} e^{\mathcal{A}_kT_k},\). Hence, the proof follows.\hfill \qed

\end{proof}
  \begin{proposition}\label{2 snap}
Consider the digraph \(\mathcal{G}_T(\mathcal{A}_i,\mathcal{B})\) associated with the linear temporally switching network described by \eqref{sys1}. Let \(\mathcal{G}_{s(1)}\) and \(\mathcal{G}_{s(2)}\) denote the signed layered graphs associated with \((\mathcal{A}_1,\mathcal{B})\) and \((\mathcal{A}_2,\mathcal{B})\) over the intervals \([t_0,t_1)\) and \([t_1,t_2]\), respectively. The system in \eqref{sys1} is \(\mathcal{SS}\) herdable within two snapshots if and only if the second snapshot corresponds to a distinct realization belonging to the same class of structured matrices and possessing the same topology and sign pattern as the first snapshot.
\end{proposition}
\begin{proof}
 The proof follows along the same lines as that of Theorem~\ref{min snap}.
\end{proof}
\begin{eg}
Consider the digraph shown in Fig.~\ref{last full}(\subref{final1}). Without loss of generality, assume that all edge weights in the first snapshot are identical, so that every nonzero entry of $\mathcal{A}_1$ shares the sign pattern indicated by the digraph. Since node~5 has both a positive and a negative walk of length~3, any parametric realization of $\mathcal{A}_1$ yields a zero row for node~5 in the portion of the controllability matrix corresponding to the first snapshot. 

However, by adopting a different realization, the sign-matched nodes become herdable under that realization. In particular, since both walks contributing to the positive and negative paths of node~5 are matched, selecting either one and applying a different realization ensures that the node becomes herdable in the second snapshot. Consequently, all nodes that are not path-sign–matched in the first snapshot become matched in the second. Therefore, the system is $\mathcal{SS}$ herdable by the end of the second snapshot.

\begin{figure}[ht]
  \centering
   \begin{subfigure}{0.18\textwidth} 
    \centering
\tikzset{every picture/.style={line width=0.55pt}} 

\begin{tikzpicture}[x=0.75pt,y=0.75pt,yscale=-1,xscale=1]

\draw    (128.71,252.43) -- (108.18,281.39) ;
\draw [shift={(116.94,269.03)}, rotate = 305.33] [fill={rgb, 255:red, 0; green, 0; blue, 0 }  ][line width=0.08]  [draw opacity=0] (7.14,-3.43) -- (0,0) -- (7.14,3.43) -- (4.74,0) -- cycle    ;
\draw  [fill={rgb, 255:red, 6; green, 80; blue, 250 }  ,fill opacity=1 ] (126.08,252.43) .. controls (126.08,250.98) and (127.26,249.81) .. (128.71,249.81) .. controls (130.17,249.81) and (131.35,250.98) .. (131.35,252.43) .. controls (131.35,253.87) and (130.17,255.04) .. (128.71,255.04) .. controls (127.26,255.04) and (126.08,253.87) .. (126.08,252.43) -- cycle ;
\draw [color={rgb, 255:red, 161; green, 27; blue, 43 }  ,draw opacity=1 ]   (107.42,200.75) .. controls (109.08,202.42) and (109.08,204.08) .. (107.41,205.75) -- (107.41,207.23) -- (107.39,215.23) ;
\draw [shift={(107.39,218.23)}, rotate = 270.1] [fill={rgb, 255:red, 161; green, 27; blue, 43 }  ,fill opacity=1 ][line width=0.08]  [draw opacity=0] (7.14,-3.43) -- (0,0) -- (7.14,3.43) -- (4.74,0) -- cycle    ;
\draw  [fill={rgb, 255:red, 239; green, 43; blue, 67 }  ,fill opacity=1 ] (104.79,198.13) .. controls (104.79,196.69) and (105.97,195.51) .. (107.42,195.51) .. controls (108.87,195.51) and (110.05,196.69) .. (110.05,198.13) .. controls (110.05,199.58) and (108.87,200.75) .. (107.42,200.75) .. controls (105.97,200.75) and (104.79,199.58) .. (104.79,198.13) -- cycle ;
\draw  [color={rgb, 255:red, 64; green, 95; blue, 17 }  ,draw opacity=1 ][fill={rgb, 255:red, 85; green, 211; blue, 33 }  ,fill opacity=1 ] (107.39,219.03) .. controls (108.86,219) and (110.08,220.25) .. (110.11,221.83) .. controls (110.15,223.4) and (108.98,224.7) .. (107.51,224.73) .. controls (106.03,224.76) and (104.81,223.51) .. (104.78,221.94) .. controls (104.75,220.37) and (105.92,219.07) .. (107.39,219.03) -- cycle ;
\draw  [fill={rgb, 255:red, 6; green, 80; blue, 250 }  ,fill opacity=1 ] (87.2,253.06) .. controls (87.2,251.61) and (88.38,250.44) .. (89.83,250.44) .. controls (91.29,250.44) and (92.47,251.61) .. (92.47,253.06) .. controls (92.47,254.5) and (91.29,255.68) .. (89.83,255.68) .. controls (88.38,255.68) and (87.2,254.5) .. (87.2,253.06) -- cycle ;
\draw  [fill={rgb, 255:red, 6; green, 80; blue, 250 }  ,fill opacity=1 ] (67.17,281.02) .. controls (67.17,279.57) and (68.35,278.4) .. (69.8,278.4) .. controls (71.25,278.4) and (72.43,279.57) .. (72.43,281.02) .. controls (72.43,282.46) and (71.25,283.63) .. (69.8,283.63) .. controls (68.35,283.63) and (67.17,282.46) .. (67.17,281.02) -- cycle ;
\draw    (88.27,255.42) -- (106.18,279.77) ;
\draw [shift={(98.76,269.69)}, rotate = 233.66] [fill={rgb, 255:red, 0; green, 0; blue, 0 }  ][line width=0.08]  [draw opacity=0] (7.14,-3.43) -- (0,0) -- (7.14,3.43) -- (4.74,0) -- cycle    ;
\draw    (88.27,255.42) -- (70.8,278.4) ;
\draw [shift={(77.96,268.98)}, rotate = 307.24] [fill={rgb, 255:red, 0; green, 0; blue, 0 }  ][line width=0.08]  [draw opacity=0] (7.14,-3.43) -- (0,0) -- (7.14,3.43) -- (4.74,0) -- cycle    ;
\draw    (129.71,255.04) -- (149.32,280.06) ;
\draw [shift={(141.12,269.6)}, rotate = 231.91] [fill={rgb, 255:red, 0; green, 0; blue, 0 }  ][line width=0.08]  [draw opacity=0] (7.14,-3.43) -- (0,0) -- (7.14,3.43) -- (4.74,0) -- cycle    ;
\draw    (107.31,225.24) -- (127.61,250.34) ;
\draw [shift={(119.1,239.82)}, rotate = 231.04] [fill={rgb, 255:red, 0; green, 0; blue, 0 }  ][line width=0.08]  [draw opacity=0] (7.14,-3.43) -- (0,0) -- (7.14,3.43) -- (4.74,0) -- cycle    ;
\draw    (107.31,225.24) -- (90.47,250.06) ;
\draw [shift={(97.43,239.8)}, rotate = 304.17] [fill={rgb, 255:red, 0; green, 0; blue, 0 }  ][line width=0.08]  [draw opacity=0] (7.14,-3.43) -- (0,0) -- (7.14,3.43) -- (4.74,0) -- cycle    ;
\draw  [fill={rgb, 255:red, 6; green, 80; blue, 250 }  ,fill opacity=1 ] (105.55,281.39) .. controls (105.55,279.95) and (106.73,278.77) .. (108.18,278.77) .. controls (109.63,278.77) and (110.81,279.95) .. (110.81,281.39) .. controls (110.81,282.84) and (109.63,284.01) .. (108.18,284.01) .. controls (106.73,284.01) and (105.55,282.84) .. (105.55,281.39) -- cycle ;
\draw  [fill={rgb, 255:red, 6; green, 80; blue, 250 }  ,fill opacity=1 ] (147.69,282.68) .. controls (147.69,281.23) and (148.87,280.06) .. (150.32,280.06) .. controls (151.78,280.06) and (152.96,281.23) .. (152.96,282.68) .. controls (152.96,284.12) and (151.78,285.29) .. (150.32,285.29) .. controls (148.87,285.29) and (147.69,284.12) .. (147.69,282.68) -- cycle ;

\draw (112.54,215.03) node [anchor=north west][inner sep=0.75pt]  [font=\tiny]  {$1$};
\draw (76.94,248.49) node [anchor=north west][inner sep=0.75pt]  [font=\tiny]  {$2$};
\draw (137.32,248.19) node [anchor=north west][inner sep=0.75pt]  [font=\tiny]  {$3$};
\draw (55.16,278.95) node [anchor=north west][inner sep=0.75pt]  [font=\tiny]  {$4$};
\draw (68.6,256.06) node [anchor=north west][inner sep=0.75pt]  [font=\scriptsize]  {$+$};
\draw (97.26,254.72) node [anchor=north west][inner sep=0.75pt]  [font=\scriptsize]  {$-$};
\draw (117.7,279.55) node [anchor=north west][inner sep=0.75pt]  [font=\tiny]  {$5$};
\draw (139.7,279.31) node [anchor=north west][inner sep=0.75pt]  [font=\tiny]  {$6$};
\draw (94.61,201.85) node [anchor=north west][inner sep=0.75pt]  [font=\scriptsize]  {$+$};
\draw (87.19,228.44) node [anchor=north west][inner sep=0.75pt]  [font=\scriptsize]  {$+$};
\draw (125.23,230.13) node [anchor=north west][inner sep=0.75pt]  [font=\scriptsize]  {$-$};
\draw (110.19,253.94) node [anchor=north west][inner sep=0.75pt]  [font=\scriptsize]  {$+$};
\draw (145.23,258.63) node [anchor=north west][inner sep=0.75pt]  [font=\scriptsize]  {$-$};
\draw (100.44,183.68) node [anchor=north west][inner sep=0.75pt]  [font=\tiny]  {$u( t)$};

\end{tikzpicture}
\caption{A digraph $\mathcal{G}_3(\mathcal{A,B})$}
\label{final1}
\end{subfigure}
 \begin{subfigure}{0.14\textwidth} 
    \centering

\tikzset{every picture/.style={line width=0.65pt}} 

\begin{tikzpicture}[x=0.75pt,y=0.75pt,yscale=-1,xscale=1]

\draw  [draw opacity=0][fill={rgb, 255:red, 74; green, 74; blue, 74 }  ,fill opacity=0.18 ] (343.23,293.15) -- (434.06,293.15) -- (434.06,304.82) -- (343.23,304.82) -- cycle ;
\draw  [draw opacity=0][fill={rgb, 255:red, 74; green, 74; blue, 74 }  ,fill opacity=0.18 ] (343.23,323.15) -- (434.06,323.15) -- (434.06,334.82) -- (343.23,334.82) -- cycle ;
\draw  [draw opacity=0][fill={rgb, 255:red, 74; green, 74; blue, 74 }  ,fill opacity=0.18 ] (344.23,351.65) -- (435.06,351.65) -- (435.06,363.32) -- (344.23,363.32) -- cycle ;
\draw [color={rgb, 255:red, 0; green, 0; blue, 0 }  ,draw opacity=0.28 ][fill={rgb, 255:red, 0; green, 0; blue, 0 }  ,fill opacity=0.18 ] [dash pattern={on 0.84pt off 2.51pt}]  (371.47,332.19) -- (371.5,357.33) ;
\draw [shift={(371.49,347.36)}, rotate = 269.92] [fill={rgb, 255:red, 0; green, 0; blue, 0 }  ,fill opacity=0.28 ][line width=0.08]  [draw opacity=0] (7.14,-3.43) -- (0,0) -- (7.14,3.43) -- (4.74,0) -- cycle    ;
\draw [color={rgb, 255:red, 0; green, 0; blue, 0 }  ,draw opacity=0.28 ][fill={rgb, 255:red, 0; green, 0; blue, 0 }  ,fill opacity=0.18 ] [dash pattern={on 0.84pt off 2.51pt}]  (409.91,330.19) -- (409.5,356.33) ;
\draw [shift={(409.67,345.86)}, rotate = 270.91] [fill={rgb, 255:red, 0; green, 0; blue, 0 }  ,fill opacity=0.28 ][line width=0.08]  [draw opacity=0] (7.14,-3.43) -- (0,0) -- (7.14,3.43) -- (4.74,0) -- cycle    ;
\draw  [fill={rgb, 255:red, 6; green, 80; blue, 250 }  ,fill opacity=1 ] (407.28,329.53) .. controls (407.28,328.08) and (408.46,326.91) .. (409.91,326.91) .. controls (411.37,326.91) and (412.55,328.08) .. (412.55,329.53) .. controls (412.55,330.97) and (411.37,332.14) .. (409.91,332.14) .. controls (408.46,332.14) and (407.28,330.97) .. (407.28,329.53) -- cycle ;
\draw [color={rgb, 255:red, 161; green, 27; blue, 43 }  ,draw opacity=1 ]   (388.62,277.85) .. controls (390.28,279.52) and (390.28,281.18) .. (388.61,282.85) -- (388.61,284.33) -- (388.59,292.33) ;
\draw [shift={(388.59,295.33)}, rotate = 270.1] [fill={rgb, 255:red, 161; green, 27; blue, 43 }  ,fill opacity=1 ][line width=0.08]  [draw opacity=0] (7.14,-3.43) -- (0,0) -- (7.14,3.43) -- (4.74,0) -- cycle    ;
\draw  [fill={rgb, 255:red, 239; green, 43; blue, 67 }  ,fill opacity=1 ] (385.99,275.23) .. controls (385.99,273.79) and (387.17,272.61) .. (388.62,272.61) .. controls (390.07,272.61) and (391.25,273.79) .. (391.25,275.23) .. controls (391.25,276.68) and (390.07,277.85) .. (388.62,277.85) .. controls (387.17,277.85) and (385.99,276.68) .. (385.99,275.23) -- cycle ;
\draw  [color={rgb, 255:red, 0; green, 0; blue, 0 }  ,draw opacity=1 ][fill={rgb, 255:red, 46; green, 133; blue, 10 }  ,fill opacity=1 ] (388.59,296.13) .. controls (390.06,296.1) and (391.28,297.35) .. (391.31,298.93) .. controls (391.35,300.5) and (390.18,301.8) .. (388.71,301.83) .. controls (387.23,301.86) and (386.01,300.61) .. (385.98,299.04) .. controls (385.95,297.47) and (387.12,296.17) .. (388.59,296.13) -- cycle ;
\draw  [fill={rgb, 255:red, 6; green, 80; blue, 250 }  ,fill opacity=1 ] (368.4,330.16) .. controls (368.4,328.71) and (369.58,327.54) .. (371.03,327.54) .. controls (372.49,327.54) and (373.67,328.71) .. (373.67,330.16) .. controls (373.67,331.6) and (372.49,332.78) .. (371.03,332.78) .. controls (369.58,332.78) and (368.4,331.6) .. (368.4,330.16) -- cycle ;
\draw  [fill={rgb, 255:red, 6; green, 80; blue, 250 }  ,fill opacity=1 ] (348.37,358.12) .. controls (348.37,356.67) and (349.55,355.5) .. (351,355.5) .. controls (352.45,355.5) and (353.63,356.67) .. (353.63,358.12) .. controls (353.63,359.56) and (352.45,360.73) .. (351,360.73) .. controls (349.55,360.73) and (348.37,359.56) .. (348.37,358.12) -- cycle ;
\draw [color={rgb, 255:red, 146; green, 17; blue, 33 }  ,draw opacity=1 ]   (369.47,332.52) -- (352,355.5) ;
\draw [shift={(359.16,346.08)}, rotate = 307.24] [fill={rgb, 255:red, 146; green, 17; blue, 33 }  ,fill opacity=1 ][line width=0.08]  [draw opacity=0] (7.14,-3.43) -- (0,0) -- (7.14,3.43) -- (4.74,0) -- cycle    ;
\draw [color={rgb, 255:red, 146; green, 17; blue, 33 }  ,draw opacity=1 ]   (410.91,332.14) -- (429.17,356.17) ;
\draw [shift={(421.61,346.23)}, rotate = 232.77] [fill={rgb, 255:red, 146; green, 17; blue, 33 }  ,fill opacity=1 ][line width=0.08]  [draw opacity=0] (7.14,-3.43) -- (0,0) -- (7.14,3.43) -- (4.74,0) -- cycle    ;
\draw [color={rgb, 255:red, 0; green, 0; blue, 0 }  ,draw opacity=0.28 ][fill={rgb, 255:red, 0; green, 0; blue, 0 }  ,fill opacity=0.18 ] [dash pattern={on 0.84pt off 2.51pt}]  (388.51,303.01) -- (408.81,328.11) ;
\draw [shift={(400.3,317.58)}, rotate = 231.04] [fill={rgb, 255:red, 0; green, 0; blue, 0 }  ,fill opacity=0.28 ][line width=0.08]  [draw opacity=0] (7.14,-3.43) -- (0,0) -- (7.14,3.43) -- (4.74,0) -- cycle    ;
\draw [color={rgb, 255:red, 146; green, 17; blue, 33 }  ,draw opacity=1 ]   (388.51,302.34) -- (371.67,327.16) ;
\draw [shift={(378.63,316.9)}, rotate = 304.17] [fill={rgb, 255:red, 146; green, 17; blue, 33 }  ,fill opacity=1 ][line width=0.08]  [draw opacity=0] (7.14,-3.43) -- (0,0) -- (7.14,3.43) -- (4.74,0) -- cycle    ;
\draw  [fill={rgb, 255:red, 6; green, 80; blue, 250 }  ,fill opacity=1 ] (406.87,358.28) .. controls (406.87,356.84) and (408.05,355.67) .. (409.5,355.67) .. controls (410.95,355.67) and (412.13,356.84) .. (412.13,358.28) .. controls (412.13,359.73) and (410.95,360.9) .. (409.5,360.9) .. controls (408.05,360.9) and (406.87,359.73) .. (406.87,358.28) -- cycle ;
\draw  [fill={rgb, 255:red, 6; green, 80; blue, 250 }  ,fill opacity=1 ] (427.89,358.78) .. controls (427.89,357.33) and (429.07,356.16) .. (430.52,356.16) .. controls (431.98,356.16) and (433.16,357.33) .. (433.16,358.78) .. controls (433.16,360.22) and (431.98,361.39) .. (430.52,361.39) .. controls (429.07,361.39) and (427.89,360.22) .. (427.89,358.78) -- cycle ;
\draw  [fill={rgb, 255:red, 6; green, 80; blue, 250 }  ,fill opacity=1 ] (368.87,359.28) .. controls (368.87,357.84) and (370.05,356.67) .. (371.5,356.67) .. controls (372.95,356.67) and (374.13,357.84) .. (374.13,359.28) .. controls (374.13,360.73) and (372.95,361.9) .. (371.5,361.9) .. controls (370.05,361.9) and (368.87,360.73) .. (368.87,359.28) -- cycle ;

\draw (393.74,292.13) node [anchor=north west][inner sep=0.75pt]  [font=\tiny]  {$1$};
\draw (358.14,325.59) node [anchor=north west][inner sep=0.75pt]  [font=\tiny]  {$2$};
\draw (418.52,325.29) node [anchor=north west][inner sep=0.75pt]  [font=\tiny]  {$3$};
\draw (336.36,356.05) node [anchor=north west][inner sep=0.75pt]  [font=\tiny]  {$4$};
\draw (349.8,333.16) node [anchor=north west][inner sep=0.75pt]  [font=\scriptsize]  {$+$};
\draw (378.46,334.82) node [anchor=north west][inner sep=0.75pt]  [font=\scriptsize]  {$-$};
\draw (374.9,356.15) node [anchor=north west][inner sep=0.75pt]  [font=\tiny]  {$5$};
\draw (437.9,357.41) node [anchor=north west][inner sep=0.75pt]  [font=\tiny]  {$6$};
\draw (375.81,278.95) node [anchor=north west][inner sep=0.75pt]  [font=\scriptsize]  {$+$};
\draw (368.39,305.54) node [anchor=north west][inner sep=0.75pt]  [font=\scriptsize]  {$+$};
\draw (406.43,307.23) node [anchor=north west][inner sep=0.75pt]  [font=\scriptsize]  {$-$};
\draw (395.39,334.54) node [anchor=north west][inner sep=0.75pt]  [font=\scriptsize]  {$+$};
\draw (398.9,355.65) node [anchor=north west][inner sep=0.75pt]  [font=\tiny]  {$5$};
\draw (427.43,336.23) node [anchor=north west][inner sep=0.75pt]  [font=\scriptsize]  {$-$};
\draw (379.64,259.67) node [anchor=north west][inner sep=0.75pt]  [font=\tiny]  {$u( t)$};

\end{tikzpicture}
\caption{Snapshot 1}
\label{final2}
\end{subfigure}
   \begin{subfigure}{0.14\textwidth} 
    \centering

\tikzset{every picture/.style={line width=0.65pt}} 

\begin{tikzpicture}[x=0.75pt,y=0.75pt,yscale=-1,xscale=1]

\draw  [draw opacity=0][fill={rgb, 255:red, 74; green, 74; blue, 74 }  ,fill opacity=0.18 ] (530.53,148.15) -- (621.36,148.15) -- (621.36,159.82) -- (530.53,159.82) -- cycle ;
\draw  [draw opacity=0][fill={rgb, 255:red, 74; green, 74; blue, 74 }  ,fill opacity=0.18 ] (530.53,178.15) -- (621.36,178.15) -- (621.36,189.82) -- (530.53,189.82) -- cycle ;
\draw  [draw opacity=0][fill={rgb, 255:red, 74; green, 74; blue, 74 }  ,fill opacity=0.18 ] (531.53,206.65) -- (622.36,206.65) -- (622.36,218.32) -- (531.53,218.32) -- cycle ;
\draw [color={rgb, 255:red, 139; green, 27; blue, 42 }  ,draw opacity=1 ]   (558.77,186.52) -- (558.8,211.67) ;
\draw [shift={(558.79,201.69)}, rotate = 269.92] [fill={rgb, 255:red, 139; green, 27; blue, 42 }  ,fill opacity=1 ][line width=0.08]  [draw opacity=0] (7.14,-3.43) -- (0,0) -- (7.14,3.43) -- (4.74,0) -- cycle    ;
\draw [color={rgb, 255:red, 139; green, 27; blue, 42 }  ,draw opacity=1 ]   (597.21,184.53) -- (596.8,210.67) ;
\draw [shift={(596.97,200.2)}, rotate = 270.91] [fill={rgb, 255:red, 139; green, 27; blue, 42 }  ,fill opacity=1 ][line width=0.08]  [draw opacity=0] (7.14,-3.43) -- (0,0) -- (7.14,3.43) -- (4.74,0) -- cycle    ;
\draw  [fill={rgb, 255:red, 6; green, 80; blue, 250 }  ,fill opacity=1 ] (594.58,184.53) .. controls (594.58,183.08) and (595.76,181.91) .. (597.21,181.91) .. controls (598.67,181.91) and (599.85,183.08) .. (599.85,184.53) .. controls (599.85,185.97) and (598.67,187.14) .. (597.21,187.14) .. controls (595.76,187.14) and (594.58,185.97) .. (594.58,184.53) -- cycle ;
\draw [color={rgb, 255:red, 161; green, 27; blue, 43 }  ,draw opacity=1 ]   (575.92,132.85) .. controls (577.58,134.52) and (577.58,136.18) .. (575.91,137.85) -- (575.91,139.33) -- (575.89,147.33) ;
\draw [shift={(575.89,150.33)}, rotate = 270.1] [fill={rgb, 255:red, 161; green, 27; blue, 43 }  ,fill opacity=1 ][line width=0.08]  [draw opacity=0] (7.14,-3.43) -- (0,0) -- (7.14,3.43) -- (4.74,0) -- cycle    ;
\draw  [fill={rgb, 255:red, 239; green, 43; blue, 67 }  ,fill opacity=1 ] (573.29,130.23) .. controls (573.29,128.79) and (574.47,127.61) .. (575.92,127.61) .. controls (577.37,127.61) and (578.55,128.79) .. (578.55,130.23) .. controls (578.55,131.68) and (577.37,132.85) .. (575.92,132.85) .. controls (574.47,132.85) and (573.29,131.68) .. (573.29,130.23) -- cycle ;
\draw  [color={rgb, 255:red, 0; green, 0; blue, 0 }  ,draw opacity=1 ][fill={rgb, 255:red, 46; green, 133; blue, 10 }  ,fill opacity=1 ] (575.89,151.13) .. controls (577.36,151.1) and (578.58,152.35) .. (578.61,153.93) .. controls (578.65,155.5) and (577.48,156.8) .. (576.01,156.83) .. controls (574.53,156.86) and (573.31,155.61) .. (573.28,154.04) .. controls (573.25,152.47) and (574.42,151.17) .. (575.89,151.13) -- cycle ;
\draw  [fill={rgb, 255:red, 6; green, 80; blue, 250 }  ,fill opacity=1 ] (555.7,185.16) .. controls (555.7,183.71) and (556.88,182.54) .. (558.33,182.54) .. controls (559.79,182.54) and (560.97,183.71) .. (560.97,185.16) .. controls (560.97,186.6) and (559.79,187.78) .. (558.33,187.78) .. controls (556.88,187.78) and (555.7,186.6) .. (555.7,185.16) -- cycle ;
\draw  [fill={rgb, 255:red, 6; green, 80; blue, 250 }  ,fill opacity=1 ] (535.67,213.12) .. controls (535.67,211.67) and (536.85,210.5) .. (538.3,210.5) .. controls (539.75,210.5) and (540.93,211.67) .. (540.93,213.12) .. controls (540.93,214.56) and (539.75,215.73) .. (538.3,215.73) .. controls (536.85,215.73) and (535.67,214.56) .. (535.67,213.12) -- cycle ;
\draw [color={rgb, 255:red, 0; green, 0; blue, 0 }  ,draw opacity=0.31 ]   (556.77,187.52) -- (539.3,210.5) ;
\draw [shift={(546.46,201.08)}, rotate = 307.24] [fill={rgb, 255:red, 0; green, 0; blue, 0 }  ,fill opacity=0.31 ][line width=0.08]  [draw opacity=0] (7.14,-3.43) -- (0,0) -- (7.14,3.43) -- (4.74,0) -- cycle    ;
\draw [color={rgb, 255:red, 0; green, 0; blue, 0 }  ,draw opacity=0.31 ]   (598.21,187.14) -- (616.47,211.17) ;
\draw [shift={(608.91,201.23)}, rotate = 232.77] [fill={rgb, 255:red, 0; green, 0; blue, 0 }  ,fill opacity=0.31 ][line width=0.08]  [draw opacity=0] (7.14,-3.43) -- (0,0) -- (7.14,3.43) -- (4.74,0) -- cycle    ;
\draw [color={rgb, 255:red, 139; green, 27; blue, 42 }  ,draw opacity=1 ]   (575.81,157.34) -- (596.11,182.44) ;
\draw [shift={(587.6,171.92)}, rotate = 231.04] [fill={rgb, 255:red, 139; green, 27; blue, 42 }  ,fill opacity=1 ][line width=0.08]  [draw opacity=0] (7.14,-3.43) -- (0,0) -- (7.14,3.43) -- (4.74,0) -- cycle    ;
\draw [color={rgb, 255:red, 0; green, 0; blue, 0 }  ,draw opacity=0.31 ]   (575.81,157.34) -- (558.97,182.16) ;
\draw [shift={(565.93,171.9)}, rotate = 304.17] [fill={rgb, 255:red, 0; green, 0; blue, 0 }  ,fill opacity=0.31 ][line width=0.08]  [draw opacity=0] (7.14,-3.43) -- (0,0) -- (7.14,3.43) -- (4.74,0) -- cycle    ;
\draw  [fill={rgb, 255:red, 6; green, 80; blue, 250 }  ,fill opacity=1 ] (594.17,213.28) .. controls (594.17,211.84) and (595.35,210.67) .. (596.8,210.67) .. controls (598.25,210.67) and (599.43,211.84) .. (599.43,213.28) .. controls (599.43,214.73) and (598.25,215.9) .. (596.8,215.9) .. controls (595.35,215.9) and (594.17,214.73) .. (594.17,213.28) -- cycle ;
\draw  [fill={rgb, 255:red, 6; green, 80; blue, 250 }  ,fill opacity=1 ] (615.19,213.78) .. controls (615.19,212.33) and (616.37,211.16) .. (617.82,211.16) .. controls (619.28,211.16) and (620.46,212.33) .. (620.46,213.78) .. controls (620.46,215.22) and (619.28,216.39) .. (617.82,216.39) .. controls (616.37,216.39) and (615.19,215.22) .. (615.19,213.78) -- cycle ;
\draw  [fill={rgb, 255:red, 6; green, 80; blue, 250 }  ,fill opacity=1 ] (556.17,214.28) .. controls (556.17,212.84) and (557.35,211.67) .. (558.8,211.67) .. controls (560.25,211.67) and (561.43,212.84) .. (561.43,214.28) .. controls (561.43,215.73) and (560.25,216.9) .. (558.8,216.9) .. controls (557.35,216.9) and (556.17,215.73) .. (556.17,214.28) -- cycle ;

\draw (581.04,147.13) node [anchor=north west][inner sep=0.75pt]  [font=\tiny]  {$1$};
\draw (545.44,180.59) node [anchor=north west][inner sep=0.75pt]  [font=\tiny]  {$2$};
\draw (605.82,181.29) node [anchor=north west][inner sep=0.75pt]  [font=\tiny]  {$3$};
\draw (523.66,211.05) node [anchor=north west][inner sep=0.75pt]  [font=\tiny]  {$4$};
\draw (537.1,188.16) node [anchor=north west][inner sep=0.75pt]  [font=\scriptsize]  {$+$};
\draw (565.76,189.82) node [anchor=north west][inner sep=0.75pt]  [font=\scriptsize]  {$-$};
\draw (562.2,211.15) node [anchor=north west][inner sep=0.75pt]  [font=\tiny]  {$5$};
\draw (624.36,210.05) node [anchor=north west][inner sep=0.75pt]  [font=\tiny]  {$6$};
\draw (563.11,133.95) node [anchor=north west][inner sep=0.75pt]  [font=\scriptsize]  {$+$};
\draw (555.69,160.54) node [anchor=north west][inner sep=0.75pt]  [font=\scriptsize]  {$+$};
\draw (593.73,162.23) node [anchor=north west][inner sep=0.75pt]  [font=\scriptsize]  {$-$};
\draw (582.69,189.54) node [anchor=north west][inner sep=0.75pt]  [font=\scriptsize]  {$+$};
\draw (586.2,210.65) node [anchor=north west][inner sep=0.75pt]  [font=\tiny]  {$5$};
\draw (615.23,191.73) node [anchor=north west][inner sep=0.75pt]  [font=\scriptsize]  {$-$};
\draw (569.14,117.17) node [anchor=north west][inner sep=0.75pt]  [font=\tiny]  {$u( t)$};

\end{tikzpicture}
\caption{Snapshot 2}
\label{final3}
\end{subfigure}
\caption{The temporally switching digraph becomes $\mathcal{SS}$ herdable within two snapshots, as all nodes are path-sign–matched across the two snapshots.}
\label{last full}
\end{figure}
  \end{eg}

\section{CONCLUSION AND FUTURE DIRECTION}
\vspace{-2mm}
This paper investigates the $\mathcal{SS}$ herdability of temporally switching networks with fixed topology, where the underlying structure is directed. We show that such digraphs can achieve complete $\mathcal{SS}$ herdability even in the presence of signed and layered dilations conditions under which static networks fail to be herdable. By employing a relaxed form of sign matching, referred to as path–sign matching, we demonstrate that complete $\mathcal{SS}$ herdability can be achieved within two snapshots for structurally similar temporal digraphs. Future work will explore extending these results to temporal networks without topological constraints, as well as investigating multi-leader $\mathcal{SS}$ herdability.

\addtolength{\textheight}{-12cm}   



\section*{AI Use Declaration}
ChatGPT was used solely to improve the language and readability of the manuscript. 
The author remains fully responsible for the accuracy and integrity of the work.

\bibliography{ifacconf}             
                                                   







\appendix
\end{document}